\documentclass[11pt,letterpaper]{article}

\usepackage{amsfonts}
\usepackage{amsmath}
\usepackage{amssymb}
\usepackage{amsthm}
\usepackage{multirow}
\usepackage{fullpage}

\usepackage{graphicx}
\graphicspath{{./figures/}}

\usepackage[colorlinks,linktocpage=true,pagebackref=true]{hyperref}
\usepackage{url}

\usepackage{scribe}

\newtheorem{theorem}{Theorem}
\newtheorem{lemma}[theorem]{Lemma}

\newtheorem{corollary}[theorem]{Corollary}

\theoremstyle{remark}

\renewcommand{\setminus}{-}
\newcommand{\first}{\mathsf{first}}
\newcommand{\last}{\mathsf{last}}

\newcommand{\comment}[1]{}

\title{Finding All Useless Arcs in Directed Planar Graphs}
\author{
  Jittat Fakcharoenphol\thanks{
    Kasetsart University, Bangkok, Thailand}
  \and
  Bundit Laekhanukit\thanks{
    Max-Planck-Institut f\"ur Informatik, Saarbr\"{u}ken, Germany
    \&
    Shanghai University of Finance and Economics, China.
    Partially supported by the ISF grant \#621/12. I-CORE grant \#4/11 and NSF \#CCF-1740425.}
  \and
  Pattara Sukprasert\thanks{
    Kasetsart University, Thailand \&
    University of Maryland, College Park MD, USA}
}

\begin{document}

\maketitle

\begin{abstract}
%
We present a linear-time algorithm for
simplifying flow networks on directed planar graphs:
Given a directed planar graph on $n$ vertices, a source vertex $s$ and
a sink vertex $t$,
our algorithm removes all the arcs
that do not participate in any simple $s,t$-path in linear-time.
The output graph produced by our algorithm satisfies the
prerequisite needed by
the $O(n\log n)$-time algorithm of Weihe [FOCS'94 \& JCSS'97]
for computing maximum $s,t$-flow in directed planar graphs.
Previously, Weihe's algorithm could not run in $O(n\log n)$-time
due to the absence of the preprocessing step;
all the preceding algorithms run in $\tilde{\Omega}(n^2)$-time
[Misiolek-Chen, COCOON'05 \& IPL'06; Biedl, Brejov{\'{a}} and Vinar, MFCS'00].
Consequently, this provides an alternative $O(n\log n)$-time algorithm
for computing maximum $s,t$-flow in directed planar graphs
in addition to the known $O(n\log n)$-time algorithms
[Borradaile-Klein, SODA'06 \& J.ACM'09;
 Erickson, SODA'10].

Our algorithm can be seen as a (truly) linear-time $s,t$-flow sparsifier
for directed planar graphs, which runs faster than
any maximum $s,t$-flow algorithm
(which can also be seen of as a sparsifier).
The simplified structures of the resulting graph might be useful
in future developments of maximum $s,t$-flow algorithms
in both directed and undirected planar graphs.
\end{abstract}


\section{Introduction}
\label{sec:intro}

The {\em maximum $s,t$-flow} problem is a fundamental problem in
Combinatorial Optimization that has numerous applications
in both theory and practice.
A basic instance of maximum flow where the underlying graph
is planar has been considered as an important special case
and has been studied since 50's in the early work of Ford and
Fulkerson \cite{FordF56}.
Since then, there have been steady developments of maximum flow algorithms
on undirected planar graphs.
Itai and Shiloach \cite{ItaiS79} proposed an algorithm for
the maximum $s,t$-flow problem on undirected planar graphs
that runs in $O(n^2\log n)$ time,
and in subsequent works \cite{Hassin81,Reif83,HassinJ85,Frederickson87,KleinRRS94-STOC94,HenzingerKRS97,ItalianoNSW11}, the running time have been
improved to the current best $O(n\log\log{n})$-time algorithm
by Italiano~et~al.~\cite{ItalianoNSW11}.

Another line of research is the study of the maximum $s,t$-flow
problem in directed planar graphs. The fastest algorithm with the
running time of $O(n\log n)$ is due to Borradaile and
Klein~\cite{BorradaileK09}.  Historically, in 1994,
Weihe~\cite{Weihe97} presented a novel approach that would solve the
maximum $s,t$-flow problem on directed planar graphs in $O(n\log n)$
time.  However, Weihe's algorithm requires a preprocessing step that
transforms an input graph into a particular form: (1) each vertex
(except source and sink) has degree three, (2) the planar
embedding has no clockwise cycle, and (3) every arc must participate
in some {\em simple $s,t$-path}.

The first condition can be guaranteed by a simple
(and thus linear-time) reduction,
and the second condition can be obtained by an algorithm of
Khuller, Naor and Klein in \cite{KhullerNK93} which runs in
$O(n\log n)$-time 
(this was later improved to linear time \cite{HenzingerKRS97}).
Unfortunately, for the third condition, there was no known algorithm
that could remove all such {\em useless} arcs in $O(n\log n)$-time.
As this problem seems to be simple, this issue had not been noticed
until it was pointed out much later by Biedl, Brejov{\'{a}} and Vinar
\cite{BiedlBV00}.
Although an $O(n\log n)$-time algorithm for the maximum $s,t$-flow
problem in directed planar graphs has been devised by Borradaile and
Klein \cite{BorradaileK09}, the question of removing all the useless
arcs in $O(n\log n)$-time remains unsolved.

In this paper, we study the {\em flow network simplification} problem,
where we are given a directed planar graph $G=(V,E)$ on $n$ vertices,
a source vertex $s$ and a sink vertex $t$, and the goal is to remove
all the arcs $e$ that are not contained in any simple $s,t$-path.
One may observe that the problem of determining the usefulness of a
single arc involves finding two vertex-disjoint paths, which is
NP-complete in general directed graphs \cite{FortuneHW80}.
Thus, detecting all the useless arcs at once is non-trivial.
Here we present a linear-time algorithm that determines all the
useless arcs, thus solving the problem left open in the work of
Weihe~\cite{Weihe97} and settling the complexity of
simplifying a flow network in directed planar graphs.

The main ingredient of our algorithm is a decomposition algorithm
that slices a plane graph into small strips with simple structures.
This allows us to design an appropriate algorithm and data structure 
to handle each strip separately.
Our data structure is simple but requires a rigorous analysis of the
properties of a planar embedding.
We use information provided by the data structure to determine
whether each arc $e$ is contained in a simple $s,t$-path $P$
in $O(1)$ time, thus yielding a linear-time algorithm.
The main difficulty is that we cannot afford to explicitly
compute such $s,t$-path $P$ (if one exists) as it would lead to
$O(n^2)$ running time.
The existence of any path $P$ can only be determined implicitly by
our structural lemmas.


Our main algorithm runs in linear time.
However, it requires the planar embedding
(which is assumed to be given together with the input graph)
to contain no clockwise cycle
and that every vertex except the source and the sink has
degree exactly three.
We provide a sequence of linear-time reductions
that produces a plane graph that satisfies the above requirements
with the guarantees that
the value of maximum $s,t$-flow is preserved.
In particular, we apply a standard degree-reduction to reduce
the degree of the input graph and invoke
the algorithm of
Khuller, Naor, and Klein \cite{KhullerNK93} to modify the planar
embedding using an application of shortest path computation in the
dual graph, which can be done in linear time by the algorithm of
Henzinger, Klein, Rao, and Subramanian~\cite{HenzingerKRS97}.

\subsection{Our Contributions}
\label{sec:contributions}

Our contribution is two-fold.
Firstly, our algorithm removes all the useless arcs in a
directed planar $s,t$-flow network in linear-time,
which thus completes the framework of Weihe \cite{Weihe97}
and yields an alternative $O(n\log n)$-time algorithm for
computing maximum $s,t$-flow on directed planar graphs.
Secondly, our algorithm can be seen as
a (truly) linear-time $s,t$-flow sparsifier,
which runs faster than any known maximum $s,t$-flow algorithm
on directed (and also undirected) planar graphs
(which can be seen as an $s,t$-flow sparsifier as well).
%
%
The plane graph produced by our algorithm has simple structures
that can be potentially exploited in the further developments of
network flow algorithms;
in particular, this may lead to a simpler or faster algorithm
for computing maximum $s,t$-flow in directed planar graphs.
In addition, our algorithm could be adapted to
simplify a flow network in undirected planar graphs as well.


\subsection{Related Work}
\label{sec:related-work}

Planar duality plays crucial roles in planar flow algorithm
developments.
Many maximum flow algorithms in planar graphs exploit
a shortest path algorithm as a subroutine.
However, the $O(n\log n)$-time algorithm for the maximum
$s,t$-flow problem in directed planar graphs by Borradaile and Klein
\cite{BorradaileK09}, which is essentially a left-most augmenting-path
algorithm, is not based on shortest path algorithms.
Erickson~\cite{Erickson10} reformulated the maximum flow problem
as a parametric shortest path problem.
The two $O(n\log n)$-algorithms were shown to be the same.
\cite{Erickson10} but with different interpretation.
Roughly speaking,
the first one by Borradaile and Klein runs mainly on the primal graph
while the latter one by Erickson runs mainly on the dual graph.
Borradaile and Harutyunyan~\cite{BorradaileH13} explored this
path-flow duality further and showed a correspondence between maximum
flows and shortest paths in directed planar graphs with no restriction
on the locations of the source and the sink.

For undirected planar unit capacity networks,
Weihe~\cite{Weihe97-undirected} presented a linear time algorithm for
the maximum $s,t$-flow.  Brandes and Wagner~\cite{BrandesW00} and
Eisenstat and Klein~\cite{EisenstatK13} presented linear time
algorithms when the planar unit-capacity network is directed.

A seminal work of Miller and Naor~\cite{MillerN95} studied flow
problems in planar graphs with multiple sources and multiple sinks.
The problem is special in planar graphs as one cannot use standard
reduction to the single-source single-sink problem without destroying
the planarity.  Using a wide range of planar graph techniques,
Borradaile, Klein, Mozes, Nussbaum,
Wulff-Nilsen~\cite{BorradaileKMNWN11} presented an $O(n\log^3 n)$
algorithm for this problem.

\section{Preliminaries and Background}
\label{sec:prelim}

We use the standard terminology in graph theory.
By the {\em degree} of a vertex in a directed graph, we mean
the sum of its indegree and its outdegree.
A path $P\subseteq G$ is a {\em simple path}
if each vertex appears in $P$ at most once.
A {\em cycle} in a directed graph is a simple directed walk
(i.e., no internal vertices appear twice)
that starts and ends at the same vertex.
A {\em strongly connected component} of $G$ is a
maximal subgraph of $G$ that is strongly connected. 
The {\em boundary} of a strongly connected component
is the set of arcs that encloses the component
(arcs on the unbounded face of the component).
We call those arcs {\em boundary arcs}.
For any strongly connected component $C$ such that $s,t\notin V(C)$,
a vertex $v\in V(C)$ is an {\em entrance} of $C$
if there exists an $s,v$-path $P\subseteq G$ such that
$P$ contains no vertices of $C$ except $v$, i.e., $V(P)\cap V(C)=\{v\}$.
Similarly, a vertex $v\in V(C)$ is an {\em exit} of $C$
if there exists an $v,t$-path $P\subseteq G$ such that
$P$ contains no vertices of $C$ except $v$, i.e., $V(P)\cap V(C)=\{v\}$.


Consider an $s,t$-flow network consisting of
a directed planar graph $G$,
a source vertex $s$ and a sink vertex $t$.
We say that an arc $uv$ is a {\em useful} arc (w.r.t. $s$ and $t$)
if there is a simple $s,t$-path $P$ containing $uv$.
Thus, the $s,u$ and the $v,t$ subpaths of $P$ have no common vertices.
Otherwise, if there is no simple $s,t$-path containing $uv$,
then we say that $uv$ is a {\em useless} arc (w.r.t. $s$ and $t$).
Similarly, a path $P$ is {\em useful} (w.r.t. $s$ and $t$)
if there is a simple $s,t$-path $Q$ that contains $P$.
Note that if a path $P$ is useful, then all the arcs of
$P$ are useful.
However, the converse is not true, i.e.,
all the arcs of $P$ are useful does not imply
that $P$ is a useful path.

Throughout the paper, we denote the input directed planar graph by $G$
and denote the number of vertices and arcs of $G$
by $n$ and $m=O(n)$, respectively.
We also assume that a planar embedding of $G$ is given as input, and
the sink vertex $t$ is on the unbounded face of the embedding.
Note that we assume that our flow networks have unit-capacities
although the algorithm works on networks with arbitrary capacities as well.


\subsection{Planar Embedding and Basic Subroutines}
\label{sec:prelim:planar-embedding}
We assume that the plane graph is given in a standard format
in which the input graph is represented by
an adjacency list whose arcs sorted in counterclockwise order.
To be precise, our input is a table $\calT$ of adjacency lists
of arcs incident to each vertex $v\in V(G)$
(i.e., $V(G)$ is the index of $\calT$)
sorted in counterclockwise order.
Each adjacency list $\calT(v)$ is a doubly linked-list
of arcs having $v$ as either heads or tails
(i.e., arcs of the forms $uv$ or $vw$).
%
Given an arc $vw$ one can find an arc $vw'$ next to $vw$ in the
counterclockwise (or clockwise) order.
This allows us to query in $O(1)$ time the {\em right-most arc}
(resp., {\em left-most arc}) of $uv$, which is an
arc $vw$ in the reverse direction that is nearest to $uv$ in the
counterclockwise (resp., clockwise) order.

\subsection{Left and Right Directions}
\label{sec:prelim:left-right}

Since we are given a planar embedding,
we may compare two paths/arcs using the notion of
{\em left to} and {\em right to}.
Let us consider a reference path $P$ and place it
so that the start vertex of $P$ is below its end vertex.
We say that a path $Q$ is {\em right to} $P$
if we can cut the plane graph along the path $P$ 
in such a way that
$Q$ is enclosed in the right side of the half-plane.
Our definition includes the case that $Q$ is a single arc.
We say that an arc $uu'$ {\em leaves $P$ to} 
(resp., {\em enters $P$ from}) the right
if $u\in V(P)$ and $uu'$ lies on the right half-plane
cutting along $P$.
Similarly, we say that a path 
$Q$ {\em leaves $P$ to}
(resp., {\em enters $P$ from}) the right
if the first arc of $Q$ leaves $P$ to (resp., enters $P$ from)
the right half-plane cutting along $P$.
These terms are also used analogously for the case of the left direction.

As mentioned, the representation of a plane graph 
described in the previous section allow us to
query in $O(1)$ time the {\em right-most} arc $vw$ of a given arc $uv$,
which is the arc in the reverse direction nearest to $uv$ in the
counterclockwise order.
Consequently, we may define the {\em right-first-search} algorithm
as a variant of the depth-first-search algorithm
that chooses to traverse to the right-most (unvisited) arc
in each step of the traversal.

We may assume that the source $s$ has a single arc $e$ leaving it.
Thus, the right-first-search (resp., left-first-search) started from $s$
gives a unique ordering of vertices
(because every path must start from the arc $e$).


\subsection{Forward and Backward Paths}
\label{sec:prelim:forward-backward}

Forward and backward paths are crucial objects that we use
to determine whether an edge is useful or useless.
Let $F$ be a reference path, which we call a {\em floor}.
We order vertices on $F$ according to their appearance on $F$
and denote such order by $\pi:V(F)\rightarrow [|V(F)|]$.
We may think that $F$ is a line that goes from left to right.

Consider any path $P$ that starts and ends at vertices on $F$.
We say that $P$ is a {\em forward path $P$ w.r.t. $F$}
if the order of vertices in $V(P)\cap V(F)$ agrees with that of $F$,
i.e., for all vertices $u,v\in V(P)\cap V(F)$,
$\pi(u) < \pi(v)$ implies that $u$ appears
before $v$ on $P$;
otherwise, we say that $P$ is a
{\em backward path $P$ w.r.t. $F$}.
In general, we will use the term forward path to mean
a minimal forward path, i.e.,
a forward path $P$ that intersects $F$ only
at its start and end vertices and share no inner vertices with $F$.
We will use the term {\em extended forward path}
to mean a non-minimal forward path.
In this sense, the ceiling $U$ is also an extended forward path.
%
%
Analogously, we use the same notions for backward paths.
A path $P$ is {\em flipped} w.r.t. $F$
if it starts from the left-side of $F$ and
ends on the right-side of $F$.
Similarly, $P$ is {\em reverse-flipped} w.r.t. $F$
if it starts from the right-side of $F$ and
ends on the left-side of $F$.
A {\em non-flipped} path is a path that lies on the left-side of $F$,
and a {\em hanging} path is a path that lies on the right-side of $F$.

\subsection{Strips and Links}
\label{sec:prelim:strips}

We are ready to define our main structural components.
A strip is formed by a {\em floor} $F$ and
a forward path $U$ w.r.t. $F$, called a {\em ceiling}.
A {\em strip} is a region enclosed by
a floor $F$ and a ceiling $U$,
and we denote the strip by $C_{U,F}$.
Observe that if we draw $F$ in such a way that $F$
is a line that goes from left to right,
then $F$ lies beneath $U$ in the planar drawing.
The two paths $U$ and $F$ form the {\em boundary} of the strip $C_{U,F}$.
A strip is {\em proper} if $U$ and $F$ are internally vertex disjoint,
i.e., they have no vertices in common except the start and the end vertices.
Generally, we use the term strip to mean a proper strip
except for the case of {\em primary strip}, which we will
define later in Section~\ref{sec:s-inside}.

Consider a proper strip $C_{U,F}$.
We call arcs in $U$ and $F$ {\em boundary arcs}
and call other arcs
(i.e., arcs in $E(C_{U,F})\setminus (E(U)\cup E(F))$)
{\em inner arcs}.
Similarly, vertices in $U$ and $F$ are called {\em boundary vertices}
and other vertices
(i.e., vertices in $V(C_{U,F})\setminus (V(U)\cup V(F))$)
are called {\em inner vertices}.
%
A path $P\subseteq C_{U,F}$ which is not entirely contained in $F$ or
$U$ is called a {\em link} in $C_{U,F}$ if its start and end vertices
are boundary vertices and all other vertices are inner vertices, i.e.,
$V(P)\cap (V(U)\cup V(F))=\{u,v\}$ where $u,v$ are the first and the
last vertices of $P$, respectively.  Observe that a link has no
boundary arcs, i.e., $E(P)\cap (E(U)\cup E(F))=\emptyset$.
%
%
A link in $C_{U,F}$ whose start and end vertices are in the floor $F$
can be classified as {\em forward} and {\em backward} in the same
manner as forward and backward paths (w.r.t. the floor $F$).
Specifically, {\em forward} (resp., {\em backward}) links are minimal
forward (resp., backward) paths.
A link, however, can start at some vertex on the floor $F$ and end at
an inner vertex of the ceiling $U$; in this case, we call it an {\em
  up-cross} path.  On the other hand, if it goes from an inner vertex
of the ceiling $U$ to some vertex on the floor $F$, we call it a {\em
  down-cross} path.  A {\em hanging path} inside a strip is defined to
be a hanging path w.r.t. the ceiling $U$.  That is, a hanging path is
a link that starts and ends at inner vertices of the ceiling $U$.
Note that a hanging path can be either forward or backward hanging path.
%
%
The classification of these links are shown
in Figure~\ref{fig:paths-in-strip}.

\begin{figure}
  \centering
  \includegraphics[width=4in]{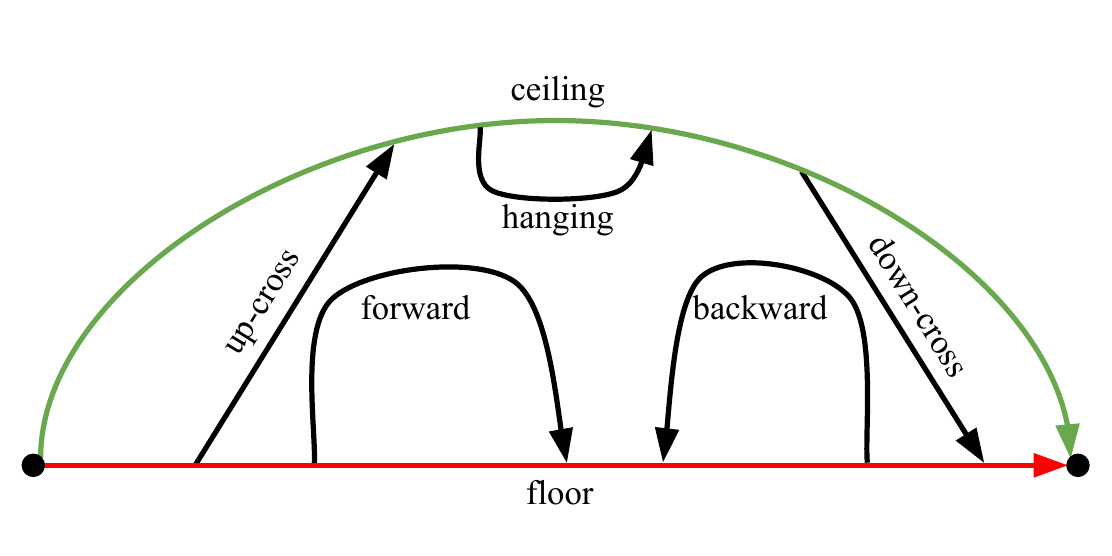}
  \caption{The classification of links in a strip.}
  \label{fig:paths-in-strip}
\end{figure}

A link inside a proper strip slices the strip into two pieces: one on
the left and on the right of the link.  If we slice a strip with a
forward, up-cross, down-cross or forward hanging path, then we have
two strips that are again proper strips.
In this sense, we say that a strip is {\em minimal} if
we cannot slice it further, i.e.,
it has neither forward, up-cross, down-cross nor forward hanging paths
(but, backward paths are allowed).

\subsection{The Usefulness of Forward Paths}

The usefulness of forward paths inside a strip $C_{U,F}$ can be determined
by the usefulness of its floor as stated in the following lemma.

\begin{lemma}
\label{lem:forward-path-is-useful}
Consider a strip $C_{U,F}$ such that the source vertex $s$
and the sink vertex $t$ are not in $C_{U,F}$.
If $F$ is useful, then so is any forward path w.r.t. $C_{U,F}$.
\end{lemma}
\begin{proof}
Let $u$ and $v$ be the start and the end vertices of $F$,
and let $P$ be a forward path connecting some vertices $q,w \in V(F)$.
Since $F$ is useful, there exists a simple $s,t$-path $R$
that contains $F$.
We may choose $R$ to be the shortest such path.
Since $s,t$ are not in $C_{U,F}$,
the minimality of $R$ implies that
none of the vertices in $V(R)\setminus V(F)$ are in $C_{U,F}$.
Consequently, the path $P$ (which is a link) has no vertices in $V(R)\setminus V(F)$.
We will construct a new simple $s,t$-path
by replacing the $q,w$-subpath of $R$ by $P$.
To be formal,
let $R_{s,q}$ and $R_{w,t}$ be the $s,q$ and $w,t$ subpaths
of $R$, respectively.
We then define a new path
$R' = R_{s,q}\cdot P\cdot R_{w,t}$.
It can be seen that $R'$ is a simple $s,t$-path.
This proves that $P$ is useful.
\end{proof}

The following corollary follows immediately from
Lemma~\ref{lem:forward-path-is-useful}.

\begin{corollary}
\label{cor:ext-forward-path-is-useful}
Consider a strip $C_{U,F}$ such that the source vertex $s$
and the sink vertex $t$ are not in $C_{U,F}$.
If $F$ is useful, then so is any extended forward path w.r.t. $C_{U,F}$.
\end{corollary}
\begin{proof}
We proceed the proof in the same way as that of
Lemma~\ref{lem:forward-path-is-useful}.
First consider any extended forward path $P$.
Observe that $P$ can be partitioned into subpaths
so that each subpath is either contained in $R$
or is a forward path w.r.t $C_{U,F}$.
We replace subpaths of $F$ by these forward paths,
which results in a new $u,v$-path $P'\supseteq P$,
where $u$ and $v$ are the start and the end vertices of the floor $F$,
respectively.
Since $F$ is useful, there exists a simple $s,t$-path $R$
that contains $F$.
We may assume that $R$ is the shortest such path.
By the minimality of $R$, we know that $P'$ has no vertices
in $V(R)\setminus V(F)$ (because $P'\subseteq C_{U,F}$).
So, we can construct a new simple $s,t$-path $R'$
by replacing the path $F\subseteq R$ in $R$ by the path $P'$.
The simple $s,t$-path $R'\supseteq P'\supseteq P$ certifies that $P$ is useful.
\end{proof}

\subsection{Other important facts}

Below are the useful facts that form the basis of our algorithm.

\begin{lemma}
\label{lem:remain-useful}
  Let $e'$ be any useless arc in $G$.
  Then any useful arc $e\in E(G)$ is also useful in $G\setminus \{e'\}$.
\end{lemma}
\begin{proof}
By the definition of useful arc,
there exists a simple $s,t$-path $P$ in $G$ containing the arc $e$,
and such path $P$ cannot contain $e'$; otherwise,
it would certify that $e'$ is useful.
Therefore, $P\subseteq G\setminus\{e'\}$, implying that $e$ is useful
in $G\setminus\{e'\}$.
\end{proof}


\begin{lemma}
\label{lem:boundary-is-ccw}
Let $C$ be a strongly connected component,
and let $B$ be the set of boundary arcs of $C$.
Then $B$ forms a counter clockwise cycle.
\end{lemma}
\begin{proof}
  Since we assume that the embedding has no clockwise cycle,
  it suffices to show that $B$ is a cycle.
  We prove by a contradiction.
  Assume that $B$ is not a cycle.
  Then we would have two consecutive arcs in $B$ that go
  in opposite directions.
  (Note that the underlying undirected graph of $B$ is always a
  cycle.)
  That is, there must exist a vertex $u$ with two leaving arcs,
  say $uv$ and $uw$.
  Since $u,v,w$ are in the same strongly connected component,
  the component $C$ must have a  $v,u$-path $P$ and and a $w,u$-path
  $P'$.
  We may assume minimality of $P$ and $P'$ and thus assume that
  they are simple paths.
  Now $P\cup\{uv\}$ and $P'\cup\{uw\}$ form a cycle,
  and only one of them can be counterclockwise
  (since $uv$ and $uw$ are in opposite direction),
  a contradiction.
\end{proof}

\begin{lemma}
\label{lem:disjoint-entrance-exit-paths}
Let $u$ and $v$ be an entrance and an exit of $C$,
and let $P_u$ and $P_v$ be an $s,u$-path $P_u$ and
a $v,t$-path $P_v$ that
contains no vertices of $C$ except $u$ and $v$, respectively.
Then $P_u$ and $P_v$ are vertex disjoint.
\end{lemma}
\begin{proof}
Suppose $P_u$ and $P_v$ are not vertex-disjoint.
Then the intersection of $P_u$ and $P_v$ induces
a strongly connected component strictly containing $C$.
This is a contradiction since $C$ is a maximal strongly connected
subgraph of $G$.
\end{proof}

\section{Overview of Our Algorithm}
\label{sec:overview}

In this section, we give an overview of the algorithm.
First, we preprocess the plane graph so that it meets the following requirements.
\begin{enumerate}
\item There is {\em no clockwise cycle} in the planar embedding.
\item The source $s$ is adjacent to only {\em one outgoing arc}.
\item Every vertex except $s$ and $t$ has {\em degree three}.
\end{enumerate}
We provide in Section~\ref{sec:preprocessing} a sequence of reductions
that outputs a graph satisfied the above conditions
while guaranteeing that
the value of maximum $s,t$-flow is preserved.

It is worth noting that the degree-three condition is not
a strict requirement for our main algorithm.
We only need every vertex to have either indegree one
or outdegree one; a network that meets this prerequisite is called
a {\em unit network}.

After the preprocessing step, we apply the main algorithm,
which first decomposes the plane graph
into a collection of strongly connected components.
The algorithm then processes each strongly connected component independently.
Notice that the usefulness of arcs that lie outside of
the strongly connected components can be determined
by basic graph search algorithms.
To see this, let us contract all the strongly connected components,
which results in a {\em directed acyclic graph} (DAG).
Since there is no cycle, an arc $uv$ is useful
if and only if $u$ is reachable from $s$, and $t$
is reachable from $v$.
So, it suffices to run any graph search algorithm
to list all the vertices reachable from $s$ and
those reachable from $t$ in the reverse graph.
Hence, we are left to work on arcs inside each strongly connected component.

We classify each strongly connected component into an {\bf outside case}
and an {\bf inside case}.
The outside case is the case that
the source $s$ lies outside of the component,
and the inside case is the case that
the source $s$ is enclosed inside the component.
(Note that since $s$ has no incoming arc, it cannot be on the
boundary of a strongly connected component.)
We deal with the outside case in Section~\ref{sec:st-outside}
and deal with the inside case in Section~\ref{sec:s-inside}.
While the inside case is more complicated, we decompose
a component of the inside case into subgraphs,
which resemblance those in the outside cases.
In both cases, the main ingredient is the {\em strip decomposition},
which allows us to handle all the useless arcs in one go.
We present the strip decomposition algorithm in Section~\ref{sec:strip-decomposition}.
Please see Figure~\ref{fig:outside-inside} for the example of outside and
inside cases.
It is possible to have many inside cases
(see Figure~\ref{fig:component-example}).

\begin{figure}
  \centering
  \begin{tabular}{cc}
  \includegraphics{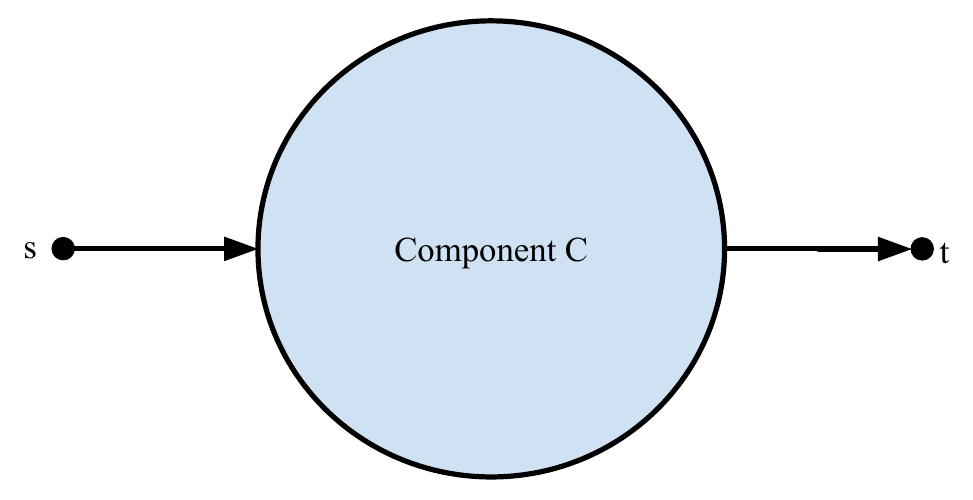}\\
  \includegraphics{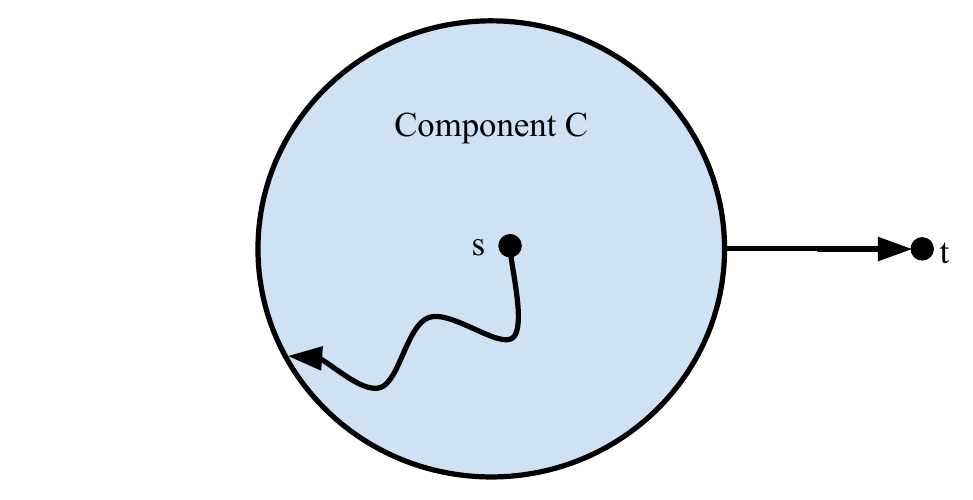}
  \end{tabular}
  \caption{The top figure illustrates the outside case, while the
    bottom figure illustrates the inside case.}
  \label{fig:outside-inside}
\end{figure}

\begin{figure}
  \centering
  \includegraphics[scale=0.75]{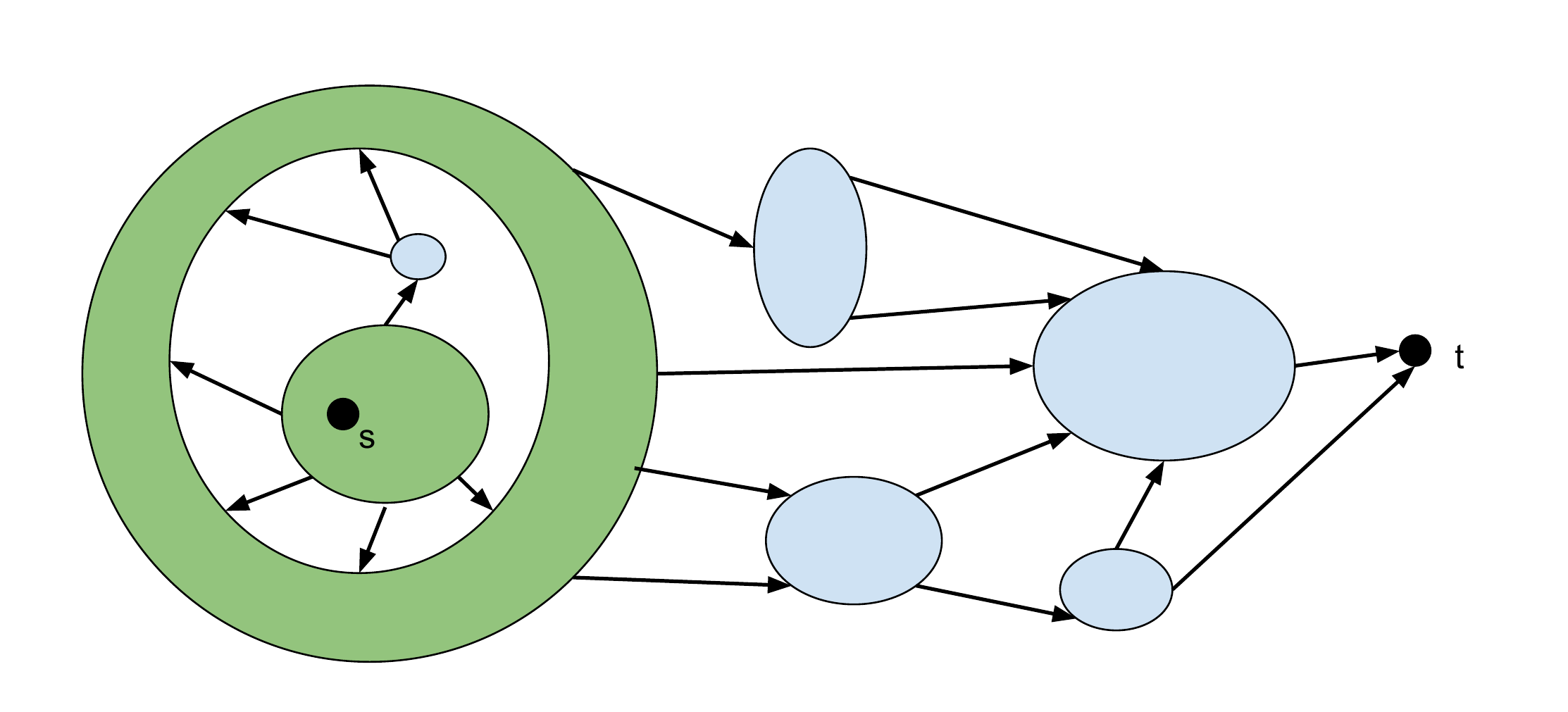}
  \caption{An example showing 7 components.  Note the nested inside cases.}
  \label{fig:component-example}
\end{figure}

\section{Preprocessing The Plane Graph}
\label{sec:preprocessing}

In this section, we present a sequence of reductions that 
provides a plane graph required by the main algorithm.
The input to this algorithm is a directed planar flow network
$(G_0,s_0,t_0)$ consisting of a plane graph $G_0$
(i.e., a directed planar graph plus a planar embedding),
a source $s_0$ and a sink $t_0$.
Our goal is to compute a directed planar flow network $(G,s,t)$
consisting of a plane graph $G$, a source $s$ and a sink $t$
with the following properties.
\begin{itemize}
\item[(1)] The value of the maximum flow in the network $(G,s,t)$
           and that in the original network $(G_0,s_0,t_0)$ are the same.
\item[(2)] The plane graph $G$ has no clockwise cycle.
\item[(3)] Every vertex in $G$ except the source $s$ and the sink $t$ has
           degree exactly three.
\item[(4)] The source $s$ has no incoming arc and has outdegree one.
\item[(5)] The resulting graph $G$ has 
           $|V(G)| \leq |E(G_0)|+2$ and $|E(G)| \leq 2|E(G_0)|+2$.
\end{itemize}

The main subroutine in the preprocessing step is 
the algorithm for removing all the clockwise cycles 
of a given plane graph $H$
by Khuller, Naor and Klein \cite{KhullerNK93},
which guarantees to preserve the value of maximum $s,t$-flow.
This subroutine can be implemented by computing 
single-source shortest path in the dual planar graph.
The setting of this shortest path problem is by placing 
a source $s^*$ on the outer-face of the dual graph,
which can be computed in linear-time using the algorithm of
Henzinger, Klein, Rao, and Subramanian~\cite{HenzingerKRS97}.
This subroutine introduces no new vertices  
but may reverse some arcs.

\subsection{Reduction}
\label{sec:preprocessing:reduction}

Our preprocessing first removes all the clockwise cycles.
Then it applies a standard reduction to reduce the degree
of all the vertices:

\begin{itemize}
\item {\bf Step 1:} We remove all the clockwise cycles from the plane graph $G_0$.
      We denote the plane graph obtained in this step by $G_1$.

\item {\bf Step 2:} We apply the {\em degree reduction} to obtain a plane graph $G_2$
      that satisfies the degree conditions in Properties (4) and (5):

      \medskip

      We replace the source $s_0$ by an arc $ss_0$, i.e.,
      we introduce the new source $s$ and add an arc $ss_0$ joining $s$ to $s_0$.
      For notation convenience, we rename $t_0$ as $t$.
      Then we replace each vertex $v$ of degree $d_v$ by a counterclockwise cycle
      $\hat{C}_v$ of length $d_v$. We re-direct each arc going to $v$ and leaving $v$ to 
      exactly one vertex $\hat{C}_v$ in such a way that no arcs cross.
      (So, each vertex of $\hat{C}_v$ corresponds to an arc incident to $v$ in $G_0$)
\end{itemize}

As we discussed, it follows by the result in \cite{KhullerNK93} that
removing clockwise cycles does not change the value of the maximum flow,
and it is not hard to see that the degree reduction has no effect 
on the value of maximum flow on edge-capacitated networks.
Thus, Property~(1) holds for $(G,s,t)$.
Property~(2) holds simply because we remove all the clockwise cycles at the end,
and Property~(3)-(5) follows directly from the degree reduction.

\section{Strip-Slicing Decomposition}
\label{sec:strip-decomposition}

The crucial part of our algorithm is the decomposition algorithm that
decomposes a strongly connected component $C$ into (proper) {\em strips},
which are regions of $C$ in the planar embedding
enclose by two arc-disjoint paths $U$ and $F$ that share
start and end vertices.

We will present the {\em strip-slicing decomposition} algorithm
that decomposes a given strip $C_{U,F}$ into a collection $\calS$ of
{\em minimal} strips in the sense that any strip in $\calS$ cannot 
be sliced into smaller strips.
Moreover, any two strips obtained by the decomposition are either
disjoint or share only their boundary vertices.
We claim that the above decomposition can be done in linear time.

If then certain prerequisites are met (we will discuss this in the later section),
then our strip-slicing decomposition gives a collection of strips such that
all the inner arcs of each strip are useless 
while all the boundary arcs are useful.

Before proceeding to the presentation of the algorithm,
we advise the readers to recall definitions 
of links and paths in a strip in Section~\ref{sec:prelim:strips}.

\subsection{The Decomposition Algorithm}
\label{sec:strip-decomposition:algorithm}

Now we describe our strip-slicing decomposition algorithm.
We start our discussion by presenting an abstraction
of the algorithm.
We denote the initial (input) strip by $C_{U^*,F^*}$.
We call $U^*$ the {\em top-most} ceiling and
call $F^*$ the {\em lowest floor}.
The top-most ceiling $U^*$ can be a dummy ceiling
that does not exist in the graph.

The decomposition algorithm taking as input a strip $C_{U,F}$.
Then it finds the ``right-most'' path $P$ w.r.t. the floor $F$
that is either a forward path or an up-cross path,
which we call a {\em slicing path}.
This path $P$ slices the strips into two regions,
the {\em up-strip} and the {\em down-strip}.
Please see Figure~\ref{fig:slicing-strip} for illustration.
Intuitively, we wish to slice the input strip into
pieces so that the final collection consists of strips 
that have no forward, up-cross, down-cross 
nor forward hanging paths
(each of these paths can slice a strip into smaller ones).

\begin{figure}
  \centering
  \begin{tabular}{l}
    \includegraphics[width=4in]{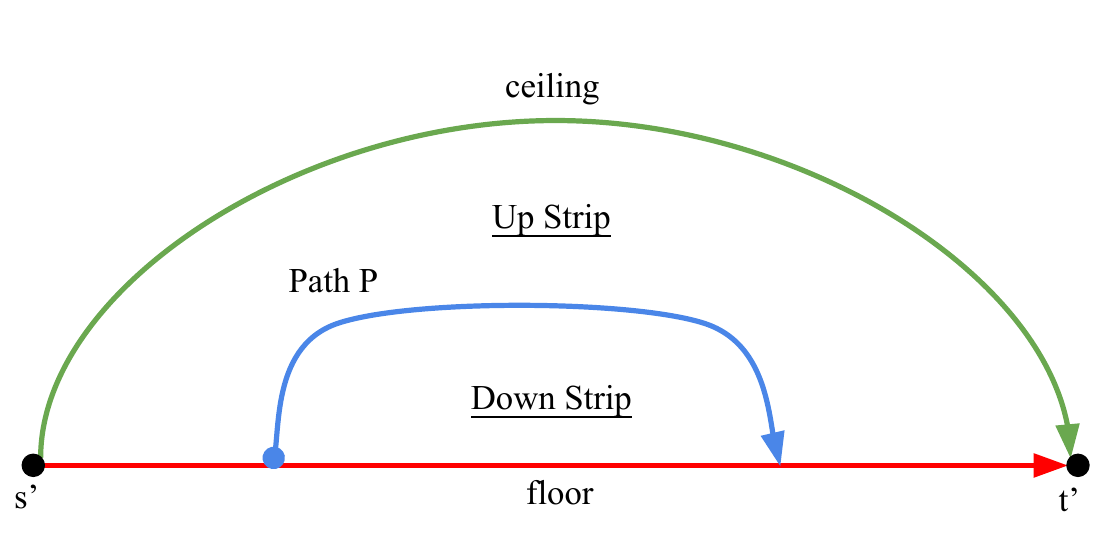}\\
    \includegraphics[width=4in]{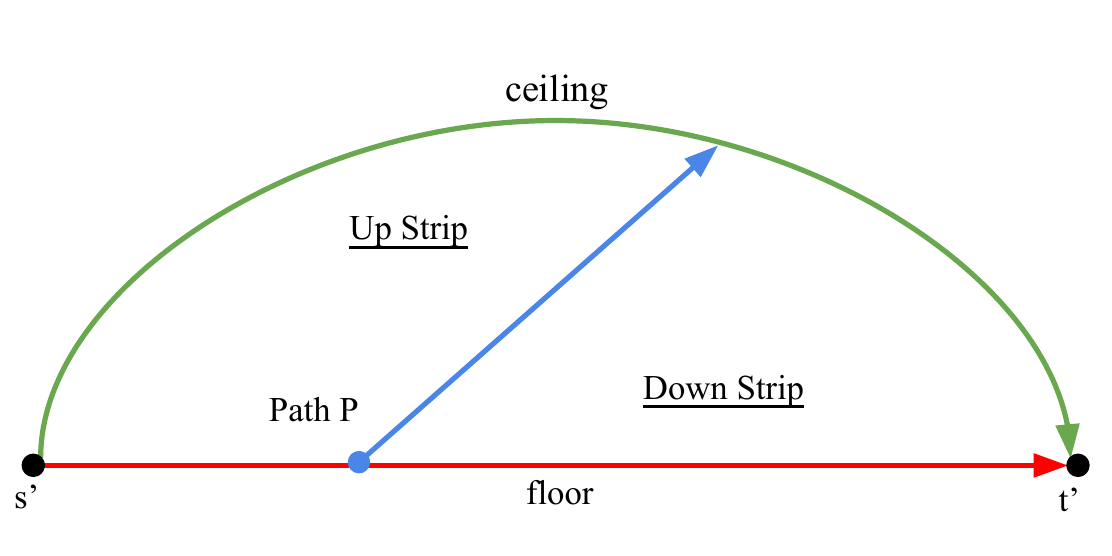}
  \end{tabular}
  \caption{The path $P$ slices the strip into up-strip and down-strips.}
  \label{fig:slicing-strip}
\end{figure}

The naive implementation yields an $O(n^2)$-time algorithm;
however, with a careful implementation, we can speed up
the running time to $O(n)$.
We remark that our decomposition algorithms (both the quadratic and linear time algorithm)
do not need the properties that the input plane graph has no clockwise cycles
and that every vertex (except source and sink) has degree three.
Thus, we are free to modify the graph as long as the graph remains planar. 

To make the readers familiar with our notations,
we start by presenting the quadratic-time algorithm in 
Section~\ref{sec:strip-decomposition:quadratic-time} and prove its correctness.
Then we will present the linear-time algorithm in 
Section~\ref{sec:strip-decomposition:linear-time} and show that 
it has the same properties as that of the quadratic-time algorithm.
The basic subroutines that are shared in both algorithms are presented
in Section~\ref{sec:strip-decomposition:subroutines}

More precisely, we prove the following lemmas in the next two sections.

\begin{lemma}
\label{lem:quadratic-time-decomposition}
There is an $O(n^2)$-time algorithm that, given as input a strip
$C_{U^*,F^*}$, outputs a collection $\calS$ of strips such that
each strip $C_{U,F}\in\calS$ has neither forward, up-cross,
down-cross nor forward hanging path (w.r.t. $U$).
%
\end{lemma}

\begin{lemma}
\label{lem:linear-time-decomposition}
There is a linear-time algorithm that, given as input a strip
$C_{U^*,F^*}$, outputs a collection $\calS$ of strips such that
each strip $C_{U,F}\in\calS$ has neither forward, up-cross,
down-cross nor forward hanging path (w.r.t. $U$).
%
\end{lemma}

The important property that we need from the decomposition algorithm is
that if the floor of the input strip $C_{U^*,F^*}$ is useful,
then all the boundary arcs of every strip $C_{U,F}\in\mathcal{S}$
are also useful.

\begin{lemma}
\label{lem:properties-strips}
Let $\mathcal{S}$ be a collection of minimal strips obtained by
running the strip-slicing decomposition on an input strip $C_{U^*,F^*}$.
Suppose the source $s$ and the sink $t$ are not enclosed in $C_{U^*,F^*}$
and that the floor $F^*$ is useful.
Then, for every strip $C_{U,F}\in\mathcal{S}$,
all the boundary arcs of $C_{U,F}$ (i.e., arcs in $U\cup F$) are useful.
\end{lemma}
\begin{proof}
We prove by induction that the decomposition algorithm maintains 
an invariant that the ceiling $U$ and the floor $F$ are useful 
in every recursive call.

For the base case, since $F^*$ is useful, 
Lemma~\ref{lem:forward-path-is-useful} implies that $U^*$ is useful
because it is a forward path w.r.t. $F^*$.
Inductively, assume that the claim holds prior to 
running the decomposition algorithm on a strip $C_{U,F}$.
Thus, $U$ and $F$ are useful.
If the algorithm finds no slicing-path, then we are done.
Otherwise, it finds a slicing-path $P$,
which slices the strip $C_{U,F}$ into
$C_{U_{up},F_{up}}$ and $C_{U_{down},F_{down}}$.
Let $s'$ and $t'$ be the start and end vertices of the strip $C_{U,F}$,
respectively.
Observe that each of the path $U_{up},F_{up},U_{down}$ and $F_{down}$
is a simple $s',t'$-path, regardless of whether the slicing $P$
is a forward path or an up-cross path.
Since $s$ and $t$ are not enclosed in $C_{U,F}$,
we can extend any simple $s',t'$-path $Q$ into a simple $s,t$-path,
which certifies that $Q$ is a useful path.
Therefore, $U_{up},F_{up},U_{down}$ and $F_{down}$
are all useful, and the lemma follows.
\end{proof}

\subsection{Basic Subroutines}
\label{sec:strip-decomposition:subroutines}

Before proceeding, we formally describe the two basic procedures that 
will be used as subroutines in the main algorithms.

\paragraph*{PROCEDURE 1: Slice the strip $C_{U,F}$ by a path $P$.}
This procedure outputs two strips $C_{U_{up},F_{up}},C_{U_{down},F_{down}}$.
We require that $P$ is either a forward-path or an up-cross path.
Let $s',t'$ be the start and end vertices of the strip $C_{U,F}$
(which are end vertices of both $U$ and $F$),
and let $u,v$ be the start and the end vertices of $P$, respectively.
We define the two new strips by defining the new floors and strips as below.
\begin{itemize}
\item {\bf Case 1}: $P$ is a  {\bf forward path}.
      Then we know that $u,v$ are on the floor $F$
      and other vertices of $P$ are not on the boundary of $C_{U,F}$. 
      To define the up-strip, we let $U_{up} := U$ and 
      construct $F_{up}$ by replacing the $u,v$-subpath of $F$ by $P$.
      To define the down-strip, we let $U_{down} := P$ and
      let $F_{down}$ be the $u,v$-subpath of $F$. 
\item {\bf Case 2}: $P$ is an  {\bf up-cross path}. 
      Then we know that $u$ is on the floor $F$ while $v$ is on the ceiling $U$
      and other vertices of $P$ has not on the boundary of $C_{U,F}$. 
      To define the up-strip, we let $U_{up}$ be the $s',u$-subpath of $U$
      and construct $F_{up}$ by concatenating 
      the $s',u$-subpath of $F$ with $P$.
      To define the down-strip, we construct $U_{down}$ by concatenating 
      the path $P$ with the $v,t'$-subpath of $U$,
      and we let $F_{down}$ be the the $u,t'$-subpath of $F$.
\end{itemize}

\paragraph*{PROCEDURE 2: Find a slicing path starting from $u\in F$.}
We require that there are no arcs leaving the strip $C_{U,F}$
(we can temporary remove them).
The algorithm runs recursively in the right first search manner
(recall this is a variant of depth first search where we choose the traverse 
to the right-most arc).
For simplicity, we assume that the first 
arc that will be scanned is the arc leaving $u$ that is next to the floor arc.
Specifically, letting $uw\in F$ be the floor arc and $uw'$ be 
the arc next to $uw$ in counterclockwise order,
the algorithm will ignore $uw$ and start by traversing $uw'$.
We terminate a recursive call when we reach a boundary vertex $v \in U\cup F$.
The algorithm then decides whether the path $P$ is a slicing path.
More precisely, if $v$ appears before $u$
(thus, $P$ is a forward path)
or $v$ is an internal vertex of the ceiling
(thus, $P$ is an up-cross path), 
then the algorithm returns that we found a slicing path $P$
plus reports that $P$ is a forward path or an up-cross path.
Otherwise, we continue the search until we visited 
all the inner vertices reachable from $u$;
in this case, the algorithm returns that 
there is no slicing path.
Notice that any slicing path $P$ found by the algorithm 
is the {\em right most} slicing path starting from $u$.

\subsection{Quadratic-Time Strip-Slicing Decomposition (Proof of Lemma~\ref{lem:quadratic-time-decomposition})}
\label{sec:strip-decomposition:quadratic-time}

\paragraph*{The Main Algorithm.}

Now we describe our $O(n^2)$-time algorithm in more details.
The algorithm reads an input strip $C_{U^*,F^*}$ and then outputs
a collection of minimal strips $\calS$.
Our algorithm consists of initialization and recursive steps
(we may think that these are two separated procedures).

Let $s',t'$ be the start and the end vertices of the input strip $C_{U^*,F^*}$
(i.e., the start and end vertices of the paths $U^*$ and $F^*$).
At the initialization step,
we remove every arc leaving the input strip
and then make sure that the {\em initial strip} has no hanging path by
adding an auxiliary ceiling $(s',x,'t')$ to enclose the input strip,
where $x'$ is a dummy vertex that does not exist in the graph.
Clearly, the modified strip has no hanging path.

Next we run the recursive procedure.
Let us denote the input strip to this procedure by $C_{U,F}$.
We first find a slicing path $P$. 
If there is no such path, 
then we add $C_{U,F}$ to the collection of strips $\calS$.
Otherwise, we found a slicing path $P$,
and we use it to slice the strip $C_{U,F}$ into 
the up-strip $C_{U_{up},F_{up}}$ and 
the down-strip $C_{U_{down},F_{down}}$.
We then make recursive calls on $C_{U_{up},F_{up}}$, $C_{U_{down},F_{down}}$,
and terminate the algorithm (or return to the parent recursive-call).

\paragraph*{Analysis}

The running-time of the algorithm can be analyzed by the following recurrent relation:
\[
T(n) = T(n_1)  + T(n - n_1) + O(n), 
\mbox{where $T(1)=1$}.
\]
This implies the quadratic running-time of the algorithm.
It is clear from the construction 
that any two strips in $\calS$ can intersect only at their boundaries
and that any strip in the collection $\calS$
contains no forward nor up-cross path.
It remains to show that the any strip at the termination has no hanging nor
down-cross path.

Observe that the assertion holds trivially at the first call to the recursive procedure
because the ceiling $U=(s',x,t')$ is an auxiliary ceiling
(we simply have no arc leaving the lone internal vertex of $U$).
We will show that this property holds inductively,
i.e., we claim that if $C_{U,F}$ has no forward hanging nor down-cross path,
then both $C_{U_{up},F_{up}}$ and $C_{U_{down},F_{down}}$ 
also have no forward hanging nor down-cross path.
We consider two cases of the slicing path $P$.

{\bf Case 1: $P$ is a forward path.} Then 
the up-strip $C_{U_{up},F_{up}}$ has no forward hanging path because $U_{up}=U$
(but, a backward hanging path may exist),
and $C_{U_{up},F_{up}}$ has no down-cross path because, 
otherwise, any such path together with a subpath of $P$
would form a down-cross path in $C_{U,F}$.

For the down-strip $C_{U_{down},F_{down}}$, 
we know that $P=F_{U_{down}}$ and that $P$ is the right-most path
w.r.t. to $F_{up}$.
This means that there is no forward hanging path w.r.t. $U_{down}$
simply because the hanging path $Q$ by definition
must lie on the right of $U_{down}=P$, which contradicts
the fact that $P$ is the right-most forward path
(we can replace a subpath of $P$ by $Q$ to
form a forward path right to $P$).
Similarly, any down-cross path w.r.t. $C_{U_{down},F_{down}}$
together with a subpath of $P$ would form a forward path right to $P$
in $C_{U,F}$, again a contradiction.

{\bf Case 2: $P$ is an up-cross path.} Then the up-strip $C_{U_{up},F_{up}}$
has no forward hanging path simply because $U_{up}$ is a subpath of $U$.
It is also not hard to see that 
$C_{U_{up},F_{up}}$ has no down-cross path.
Suppose not.
Then we have a down-cross path $Q$ that
starts from some internal vertex $x$ of $U_{up}$
and ends at some vertex $y$ in $F_{up}$.
We have two cases:
(1) if $y\in F$, then the path $Q$ is also a down-cross path in $C_{U,F}$ 
and
(2) if $y\in P$, then the path $Q$ together with a subpath of $P$
form a forward hanging path in $C_{U,F}$.
Both cases lead to contradiction.

Next consider the down-strip $C_{U_{down},F_{down}}$.
The ceiling $U_{down}$ consists of the slicing path $P$ and
a subpath of $U$.
Suppose there is a down-cross path $Q$ in $C_{U_{down},F_{down}}$.
Then we know that $Q$ must go from $P$ to $F_{down}=U$.
But, this means that we can form a forward path right to $P$ 
by concatenating a subpath of $P$ with $Q$,
which contradicts the fact that $P$ is the right-most slicing path.
Suppose there is a forward hanging path $Q$ in  $C_{U_{down},F_{down}}$.
Then it must start from some vertex $v$ in $P$
and ends on either $P$ or the subpath of $U$ in $U_{down}$.
But, both cases imply that there exists a slicing path $P'$
that lie right to $P$, which contradicts our choice of $P$.
\qed

\subsection{Linear Time Implementation}
\label{sec:strip-decomposition:linear-time}

Now we present a linear-time implementation of our
strip-slicing decomposition algorithm and thus prove
Lemma~\ref{lem:linear-time-decomposition}.
The algorithm is similar to the one with quadratic running-time
except that we keep the information of vertices that have already been scanned.
Here the recursive procedure has an additional
parameter $v^*$, which is a pointer to a vertex on the floor $F$ that 
we use as a starting point in finding the slicing path $P$.

To be precise, our algorithm again consists of two steps.
Let us denote the input strip by $C_{U^*,F^*}$
and denote its start and end vertices by $s'$ and $t'$, respectively.
At the initialization step, 
we remove all arc leaving $C_U^*,F^*$ and 
add an auxiliary ceiling $(s',x',t')$ to enclose the input strip
($x'$ is a dummy vertex).
Then we set the starting point $v^*:=s'$ and call the recursive procedure
with parameter $C_{U,F}$ (the modified strip) and $v^*$.

In the recursive procedure, 
we first order vertices on the floor $F$ by $v_1,\ldots,v_{\ell}$
(note that $v_1=s'$ and $v_{\ell}=t'$).
Let $i^*$ be the index such that $v_{i^*}=v^*$.
For each $i=i^*$ to $\ell$ in this order,
we find a slicing path starting from $v_i$
(using the right-first-search algorithm).
If we cannot find a slicing path starting from $v_i$,
then we remove all the inner vertices and arcs
scanned by the search algorithm
and iterate to the next vertex (i.e., to the vertex $v_{i+1}$).
Otherwise, we found a slicing path $P$ and exit the loop. 
Note that the loop terminates under two conditions.
Either the search algorithm finds no slicing path
or
it reports that there exists a slicing path $P$ starting from a vertex $v_i$.
For the former case, we add $C_{U,F}$ to the collection of strips $\calS$.
For the latter case, 
we remove all the inner vertices and arcs scanned by 
the search algorithm except those on $P$,
and we then slice the strip $C_{U,F}$ into 
the up-strip $C_{U_{up},F_{up}}$ and
the down-strip $C_{U_{down},F_{down}}$.
Finally, we make two calls to the recursive procedure 
with parameters $(C_{U_{up},F_{up}},v_i)$
and $(C_{U_{down},F_{down}},v_i)$, respectively.
In particular, we recursively run the procedure on the up and down strips,
and we pass the vertex $v_i$ as the starting point in finding slicing paths.

\subsubsection{Implementation Details}

Here we describe some implementation details that 
affect the running-time of our algorithm.

\paragraph*{Checking whether a vertex is in the floor or ceiling.}
Notice that checking whether a vertex is in the the floor $F$ or 
the ceiling $U$ can be an issue since the up-strip and the down-strip
share boundary vertices.
A trivial implementation would lead to storing the same vertex in $O(n)$
data structures, and the algorithm ends up running in $O(n^2)$-time.
We resolve this issue by using a coloring scheme.
We color vertices on the floor $F$ by red and 
the vertices on the ceiling $U$ by green.
(Note that we color only the internal vertices of $U$ 
since $U$ and $F$ share the start and end vertices).
We color other vertices by white.
Now we can easily check if a vertex is in the floor or ceiling by
checking its color.
When we make a recursive call to the up-strip and the down-strip.
We have to maintain the colors on the slicing path $P$
to make sure that we have a right coloring.
More precisely, when we make a recursive call to the up-strip
we color internal vertices of $P$ by red 
(since they will be vertices on the floor $F_{up}$)
and then flip their colors to green later when we make a recursive call to
the down-strip
(since they will be vertices on the ceiling $U_{down}$).
Observe that internal vertices of the slicing path $P$ are internal vertices 
of the strip $C_{U,F}$, which will be boundary vertices
of the up and down strips in the deeper recursive calls. 
As such, any vertices will change colors at most three times
(from white to green and then to red).

\paragraph*{Checking if a path is a forward path.}
Another issue is when we search for a slicing path $P$
and want to decide whether the path we found is a forward path or a backward path.
We simply solve this issue by removing all the incoming arcs of vertices 
$v_{i'}$ for all $0 \leq i'< i$.
This way any path we found must be either a forward path or  an up-cross path. 

\paragraph*{Slicing a Strip.}
Observe that we may create $O(n)$ strips
(and thus have $O(n)$ floors and ceilings)
during the run of the decomposition algorithm.
However, we cannot afford to store floors and ceilings of
every strip separately as they are not arc disjoint.
Otherwise, just storing them would require a running time of
$O(n^2)$.
To avoid this issue, we keep paths (either floor or ceiling) as
doubly linked lists.
Consequently, we can store $C_{U,F}$ by referencing to
the first and last arcs of the floor (respectively, ceiling).
Moreover, since we store paths as doubly linked list,
each {\em cut-and-join} operation can be done in $O(1)$ time.

\paragraph*{Maintaining the search to be inside $C_{U,F}$.}
As we may have an arc going leaving $C_{U,F}$,
we have to ensure that the right-first-search algorithm would not 
go outside the strip.
(Note that our algorithm does not see the real drawing of $C_{U,F}$
on a plane and thus cannot decide whether an arc is in a strip or not.)
To do so, we temporary remove the arcs leaving $C_{U,F}$ and 
put them back after the procedure terminates. 
Observe that the inner parts of up-strip $C_{U_{up},F_{up}}$ and $C_{U_{down},F_{down}}$
are disjoint.
Thus, each arc leaving $C_{U,F}$ is removed and put back at most once.

\subsubsection{Correctness}
\label{sec:strip-decomposition:linear-time:correctness}

For the correctness, it suffices the show that  
the slicing path $P$ that we found in every step is exactly the same 
as what we would have found by the quadratic-time algorithm.
We prove this by reverse induction.
The base case, where we call the procedure on the initial strip $C_{U^*,F^*}$,
holds trivially.
For the inductive step,
we assume the that both implementations found the same slicing paths so far.
Now if the recursive procedure found no slicing path, then we are done.
Otherwise, we found a slicing path $P$.
Let us order vertices on the floor $F$ by $v_1,\ldots,v_{\ell}$,
and let $v_i$ be the start vertex of $P$.
First, consider the up-strip $C_{U_{up},F_{up}}$.
Clearly, $C_{U_{up},F_{up}}$ has no up-cross path starting from
any vertex $v_{i'}$ with $0 \leq i' < i$.
We also know that $C_{U_{up},F_{up}}$ 
has no forward path $Q$ starting from $v_{i'}$;
otherwise, the path $Q$ is either a forward path in $C_{U,F}$
starting at $v_{i'}$ (if $Q$ starts and ends on $U$)
or we can form another slicing path $P'$ starting from $v_{i'}$ 
by concatenate $Q$ with a subpath of $P$.
Moreover, all the vertices that scanned by running  
the right-first-search algorithm from $v_{i}$ must be 
enclosed in the down-strip.
Thus, we conclude that removing all scanned vertices (those that are not on $P$)
does not affect the choice of slicing paths in the up-strip.

Next consider the down-strip.
In any case, we know that the slicing path $P$ 
becomes a subpath of the ceiling of $C_{U_{down},F_{down}}$.
Clearly, any path $Q$ from $P$ to $F_{down}$ would either form 
a forward path $P'$ lying right to $P$ or 
would imply that $C_{U,F}$ has a down-cross path
(this is the case when $Q$ is a subpath of a longer path 
that goes from $U$ to $F_{down}$),
which are contradictions in both cases.
Thus, we can remove every vertex that is scanned by
the right-first-search algorithm and is not on the slicing path $P$
without any effect on the choices of slicing paths in deeper recursive calls.
To conclude, our linear-time and quadratic-time implementations
have the same choices of slicing paths and must produce the same outputs.
The correctness of our algorithm is thus implied by
Lemma~\ref{lem:quadratic-time-decomposition}.

\subsubsection{Running Time Analysis}
\label{sec:strip-decomposition:linear-time:running-time-analysis}

Clearly, the initialization step of the decomposition algorithm runs
in linear-time.
Thus, it suffices to consider the running time of the recursive procedure.
We claim that each arc is scanned by our algorithm at most $O(1)$ times,
and each vertex is scanned only when we scan its arcs.
This will immediately imply the linear running time.
We prove this by induction.
First, note that we remove all the internal vertices and arcs 
that are scanned by the search algorithm
but are not on the slicing path $P$ (if it exists),
which means that these vertices and arcs are scanned only once
throughout the algorithm.
Thus, if we found no slicing path, then the 
decomposition runs in linear time on 
the number of arcs in the strip.
Otherwise, we found a slicing path $P$ starting at some vertex $v_i\in F$.
The recursive procedure slices the strip $C_{U,F}$
in such a way that the two resulting strips $C_{U_{up},F_{up}}$
and $C_{U_{down},F_{down}}$ have no inner parts in common.
Thus, we can deduce inductively that 
the running time incurred by scanning inner arcs and vertices
of the two strips are linear on the number of arcs in the strips.
Moreover, since we pass $v_i$ as the starting point of iterations 
in deeper recursive calls, we can guarantee that
any vertices $v_{i'}\in F$ with $i < i'$
(those that appear before $v_i$) will never be scanned again.
We note, however, that we may scan the vertex $v_i$ multiple times
because some arcs leaving $v_i$ that lie on the left of $P$
have never been scanned.
Nevertheless, we can guarantee that the number of times that we scan $v_i$ 
is at most the number of its outgoing arcs.
To see this, first any arc that does not lead to a slicing path will be removed.
For the arc that leads to a slicing path, we know that it will become 
an arc on the boundaries of up and down strips 
and thus will never be scanned again.
In both cases, we can see that the number of times we scan $v_i$ 
is upper bounded by the number of its outgoing arcs.
Therefore, the running time of our algorithm is linear on the number of
arcs in the input strip,
completing the proof of Lemma~\ref{lem:linear-time-decomposition}.

\qed
\section{The Outside Case: Determining useless arcs when $s$ is not
         in the component}
\label{sec:st-outside}

In this section, we describe an algorithm that determines all useful arcs of
a strongly connected component $C$ when the source $s$ is in the infinite face.

We first outline our approach.
Consider a strongly connected component $C$.
If the boundary of $C$ is not a cycle, then the strong connectivity
implies that we can extend some path on the boundary to form
a clockwise cycle, which contradicts to our initial assumption
that the embedding has no clockwise cycle.

If all arcs in the boundary of $C$, denoted by $Q$, are useless, then
no arcs of $C$ are useful.
Otherwise, we claim that some arcs of $Q$ are useful and some
are useless. Moreover, these arcs partition $Q$ into two paths
$Q_1$ and $Q_2$ such that $Q_1$ is a useful path
and all arcs of $Q_2$ are useless.

\begin{lemma}
\label{lem:one-useful}
Consider the boundary of a strongly connected component $C$.
Let $Q$ be the cycle that forms the boundary of $C$.
Then either all arcs of $C$ are useless or
there are non-trivial paths $Q_1$ and $Q_2$ such that
\begin{itemize}
  \item $E(Q)=E(Q_1)\cup E(Q_2)$.
  \item $Q_1$ is a useful path.
  \item All arcs of $Q_2$ are useless.
\end{itemize}
Moreover, there is a linear-time algorithm that computes $Q_1$ and $Q_2$.
\end{lemma}
\begin{proof}
  First, we claim that all arcs of $C$ are useless
  if and only if $Q$ has no entrance or has no exit.

  If $Q$ has no entrance or has no exit, then
  it is clear that all arcs of $C$ are useless
  because there can be no simple $s,t$-path using an edge of $C$.

  Suppose $Q$ has at least one entrance $u$ and
  at least one exit $v$.
  We will show that $Q$ has both useful and useless arcs, which
  can be partitioned into two paths as in the statement of the lemma.

  Observe that no vertices of $Q$ can be both entrance and exit
  because of the degree-three assumption.
  (Any vertex that is both entrance and exit must have degree at least
  four, two from $Q$ and two from the entrance and the exit paths.)

  For any pair of entrance and exit $u,v$,
  let $Q_{uv}$ be the $u,v$-subpath of $Q$.
  We claim that $Q_{uv}$ is a useful path.
  To see this, we apply
  Lemma~\ref{lem:disjoint-entrance-exit-paths}.
  Thus, we have an $s,u$-path $P_u$ and
  a $v,t$-path $P_v$ that are vertex disjoint.
  Moreover, $P_u$ and $P_v$ contain no vertices of $Q$ except
  $u$ and $v$, respectively.
  Thus, $R=P_u\cdot Q_{uv}\cdot P_v$ form a simple $s,t$-path,
  meaning that $Q_{uv}$ is a useful path.

  Notice that $Q$ has no useless arc only if
  $Q$ has two distinct entrances $u_1,u_2$ and two
  distinct exits $v_1,v_2$ that appear in $Q$ in interleaving
  order $(u_1,v_1,u_2,v_2)$.
  Now consider a shortest $s,u_1$-path $P_{u_1}$,
  a shortest $s,u_2$-path $P_{u_2}$,
  and the $u_1,u_2$-subpath $Q_{u_1,u_2}$ of $Q$.
  These three paths together encloses any path that leaves
  the component $C$ through an exit $v_1$.
  Thus, any $v_1,t$-path must intersect either $P_{u_1}$ or
  $P_{u_2}$, contradicting
  Lemma~\ref{lem:disjoint-entrance-exit-paths}.

  Consequently, a sequence of entrances and exits
  appear in consecutive order on $Q$.
  Let us take the first entrance $u^*$ and the last exit $v^*$.
  We know from the previous arguments that any pair of entrance
  and exit yields a useful subpath of $Q$.
  More precisely, the $u^*,v^*$-subpath $Q_1$ of $Q$ must be
  useful and must contain all the entrances and exits
  because of the choices of $u^*$ and $v^*$.
  Consider the other subpath -- the $v^*,u^*$-subpath $Q_2$ of $Q$.
  The only entrance and exit on $Q_2$ are $u^*$ and $v^*$,
  respectively.
  Thus, any path $s,t$-path $P$ that contains an arc $e$ of $Q_2$ must
  intersect with the path $Q_1$.
  But, then the path $P$ together with the boundary $Q$
  enclose the region of $C$ that has no exit.
  Thus, $P$ cannot be a simple path.
  It follows that no arcs of $Q_2$ are useful.

  To distinguish the above two cases, it suffices to compute all the
  entrances and exits of $Q$ using a standard graph searching
  algorithm. The running time of the algorithm is linear, and we need
  to apply it once for all the strongly connected components.
  This completes the proof.
\end{proof}

%

\subsection{Algorithm for The Outside Case}
\label{sec:algo-outside-case}

Now we present our algorithm for the outside case.
We assume that some arcs in the boundary of $C$ are useful and some
are useless.

Our algorithm decides whether each arc in $C$ is useful or useless
by decomposing $C$ into a collection of minimal strips
using the algorithm in Section~\ref{sec:strip-decomposition}.
Then any arc enclosed inside a minimal strip is useless
(only boundary arcs are useful).
Thus, we can determine the usefulness of arcs in each strip.
Observe, however, that the component $C$ is not a strip
because it is enclosed by a cycle instead of two paths
that go in the same direction.
We transform $C$ into a strip as follows.
First, we apply the algorithm as in Claim~\ref{lem:one-useful}.
Then we know the boundary of $C$, which is a cycle
consisting of two paths $Q_1$ and $Q_2$,
where $Q_1$ is a useful path and $Q_2$ is a useless path.
Let $F^*=Q_1$, and call it the {\em lowest-floor}.
Note that $F^*$ starts at an entrance $s_C$ and
ends at an exit $t_C$.
We add to $C$ a dummy path $U^* = (s_C,u^*,t_C)$,
where $u^*$ is dummy vertex,
and call $U^*$ the {\em top-most ceiling}, which is a dummy ceiling.
This transforms $C$ into a strip, denoted by $C_{U^*,F^*}$.
(See Figure~\ref{fig:s-outside-init-strip})

\begin{figure}
  \centering
  \includegraphics{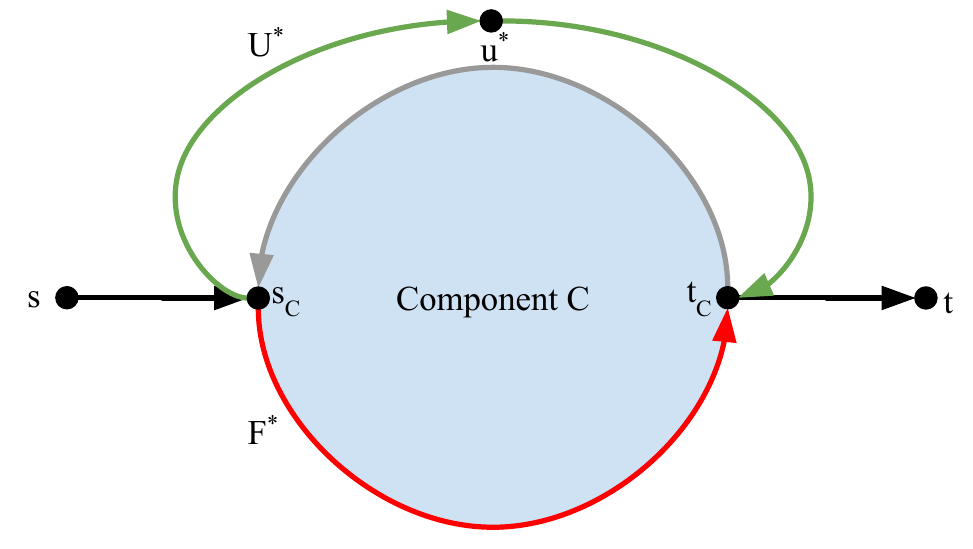}
  \caption{The initial strip before the decomposition.}
  \label{fig:s-outside-init-strip}
\end{figure}

Now we are able to decompose $C_{{U^*},{F^*}}$ into
a collection $\calS$ of minimal strips by calling
the strip-slicing decomposition on $C_{{U^*},{F^*}}$.
Since $F^*$ is a useful path, we deduce from
Corollary~\ref{cor:ext-forward-path-is-useful}
that all the strips $C_{U,F}\in \calS$ are
such that both $U$ and $F$ are useful paths.

Next we show that no inner arcs of a strip $C_{U,F}\in\calS$
are useful, which then implies only arcs on the boundaries
of strips computed by the strip-slicing decomposition are useful.
Therefore, we can determine all the useful arcs in $C$ in
linear-time.

\begin{lemma}
\label{lem:outside-no-useful-backward}
Consider a strip $C_{U,F}\in\calS$ computed by
the strip-slicing decomposition algorithm.
Then no inner arc $e$ of $C_{U,F}$ is useful.
More precisely, all arcs in $e\in E(C_{U,F})\setminus E(U\cup F)$
are useless.
\end{lemma}
\begin{proof}
Suppose for a contradiction that there is an inner arc $e$ of $C_{U,F}$
that is useful.
Then there is an $s,t$ path $P$ containing $e$.
We may also assume that $P$ contains a link $Q$.
That is, $Q$ starts and ends at some vertices $u$ and $v$ in
$V(U\cup F)$, respectively, and $Q$ contains no arcs of $E(U\cup F)$.
By Lemma~\ref{lem:linear-time-decomposition}, $Q$ must be a backward path.
Hence, it suffices to show that any backward path inside the strip
$C_{U,F}$ is useless.

By the definition of the backward path $Q$,
$v$ must appear before $u$ on $F$,
and by Lemma~\ref{lem:properties-strips}, $F$ is a useful path, meaning that
there is an $s,t$-path $R$ containing $F$.
Let $s'$ and $t'$ be the start and end vertices of
the strip $C_{U,F}$
(i.e., $s'$ and $t'$ are
the common start and end vertices of the paths $U$ and $F$,
respectively).
Then we have four cases.

\begin{itemize}
\item {\bf Case 1: $v=s'$ and $u=t'$.}
  In this case, the path $U$ and $Q$ form a clockwise cycle,
  a contradiction.

\item {\bf Case 2: $v=s'$ and $u\neq t'$.}
  In this case, $v$ has degree three in $Q\cup U\cup F$.
  Thus, the path $R\supseteq F$ must start at the same vertex as $F$;
  otherwise, $v$ would have degree at least four.
  This means that $s=s'$, but then $P$ could not be a simple
  $s,t$-path because $s$ must appear in $P$ at least twice,
  a contradiction.

\item {\bf Case 3: $v\neq s'$ and $u=t'$.}
  This case is similar to the former one.
  The vertex $u$ has degree three in $Q\cup U\cup F$.
  Thus, the path $R\supseteq F$ must end at the same vertex as $F$;
  otherwise, $u$ would have degree at least four.
  This means that $t=t'$, but then $P$ could not be a simple
  $s,t$-path because $t$ must appear in $P$ at least twice,
  a contradiction.

\item {\bf Case 4: $u\neq t'$ and $v\neq s'$.}
  Observe that $P$ must enter and leave $C_{U,F}$ on
  the right of $F$ (and thus $R$)
  at some vertices $a$ and $b$, respectively.
  We may assume wlog that the $a,u$-subpath
  and the $v,b$-subpath of $P$ are contained in $R$.
  Moreover, it can been seen that
  $u$ and $v$ appear in different orders on
  the path $R$ (that contains $F$) and
  the path $P$ (that contains $Q$).
  Thus, $P$ must intersect either
  the $s,b$-subpath of $R$, say $R_{s,b}$, or
  the $a,t$-subpath of $R$, say $R_{a,t}$.
  (Note that $R$ contains no inner vertices of $Q$
  by the definition of a backward path.)

  Suppose $P$ intersects $R_{a,t}$ at some vertex $x$,
  and assume that $x$ is the last vertex of $P$ in
  the intersection of $P$ and $R_{a,t}$.
  Since $P$ enters $R$ at the vertex $a$ from the right,
  the $x,a$-subpath of $P$ must lie on the right of $R$.
  But, then the union of the $x,a$-subpath of $P$ and
  the $a,x$-subpath of $R$ forms a clockwise cycle,
  a contradiction.
  (See Figure ~\ref{fig:outside-no-usefull-backward-4-1}
  for illustration.)

  Now suppose $P$ intersects $R_{s,b}$ at some vertex $y$,
  and assume that $y$ is the first vertex of $P$ in
  the intersection of $P$ and $R_{s,b}$.
  Since $P$ leaves $R$ at the vertex $b$ from the right,
  the $b,y$-subpath of $P$ must lie on the right of $R$.
  But, then the union of the $b,y$-subpath of $P$ and
  the $y,b$-subpath of $R$ forms a clockwise cycle,
  again a contradiction.
\end{itemize}
Therefore, all backward paths inside the strip $C_{U,F}$ are useless
and so do inner arcs of $C_{U,F}$.
\end{proof}

\begin{figure}
  \centering
  \includegraphics[width=4in]{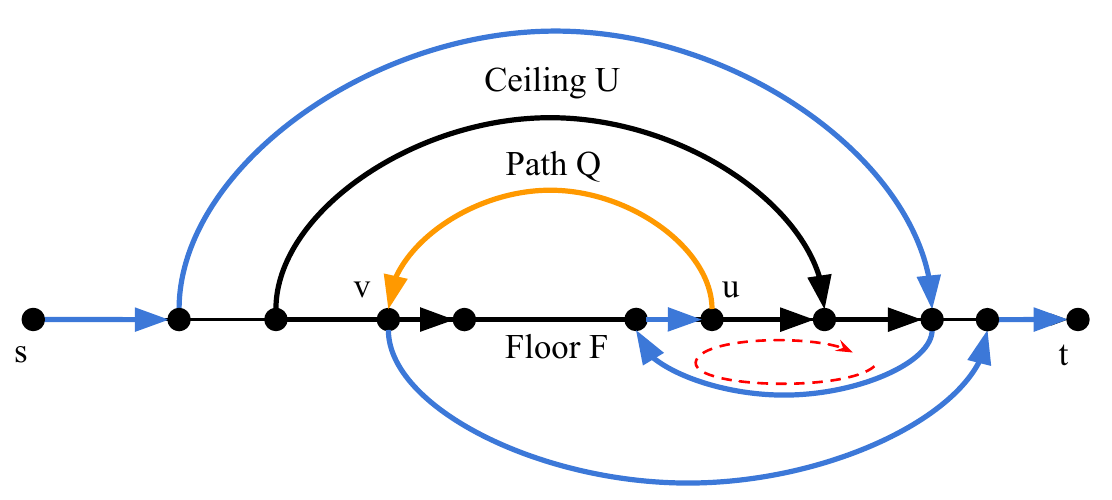}
  \caption{An illustration of Case~4 in the proof of
           Lemma~\ref{lem:outside-no-useful-backward}}
  \label{fig:outside-no-usefull-backward-4-1}
\end{figure}

\section{The Inside Case: Determining useful arcs of a strongly
  connected component when $s$ is inside of the component}
\label{sec:s-inside}

In this section, we describe an algorithm for the case that
the source vertex $s$ is in the strongly connected component $C$.
In this case, the source vertex $s$ is enclosed in some face $f_s$ in
the component $C$.  It is possible that $s$ is also enclosed
nestedly in other strongly connected components (see
Figure~\ref{fig:component-example}).  We will perform a simple
reduction so that (1) the source $s$ is enclosed inside the component $C$,
(2) $s$ has only one outgoing arc, and
(3) every vertex except $s$ (and $t$) has degree three:
First, we contract the maximal component containing $s$ inside $f_s$
into $s$.
Now the source $s$ has degree more than one in a new graph.
We thus add an arc $s's$ and declare $s'$ as a new source vertex.
Then we replace $s$ by a binary tree on $d$ leaves to maintain
the degree-three requirement.
See Figure~\ref{fig:inside-reduction}.

\begin{figure}
  \centering
  \includegraphics[width=\textwidth]{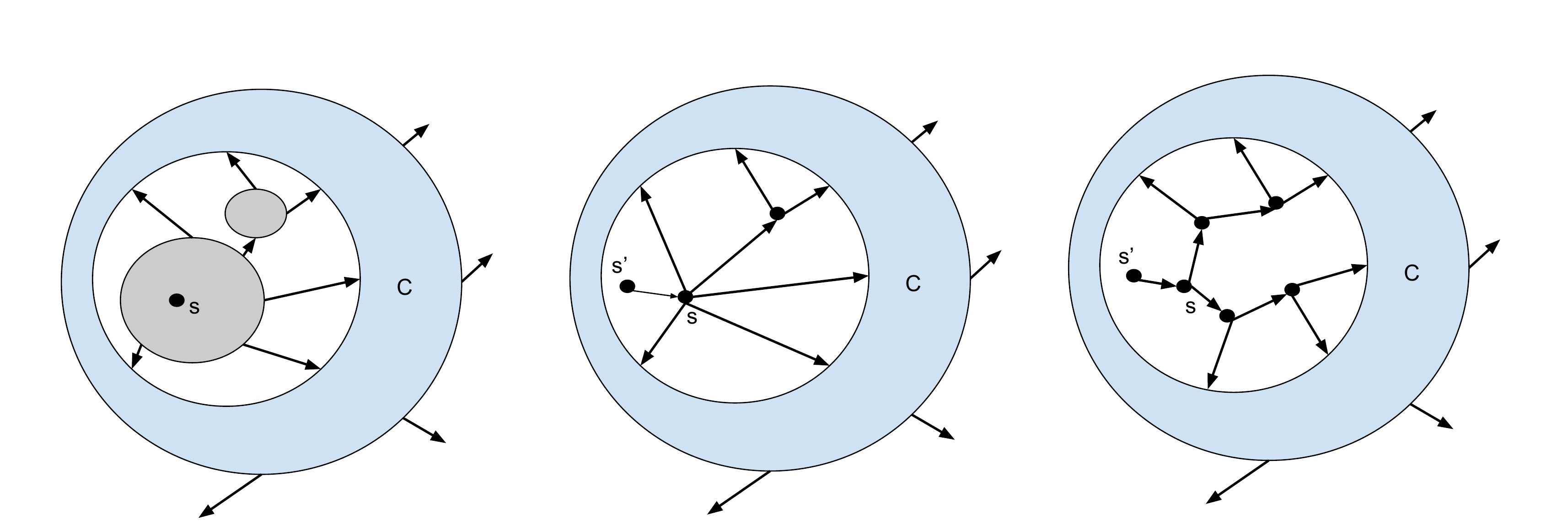}
  \caption{A reduction for simplifying the source structure for the
    inside case.}
  \label{fig:inside-reduction}
\end{figure}

\subsection{Overview}

We decompose the problem of determining the usefulness of component $C$
into many subproblems and will deal with each subproblem
(almost) independently.
The basis of our decomposition is in slicing the component with the
{\em lowest-floor path} $F^*$ and the {\em top-most ceiling} $U^*$.
These two paths then divide the component $C$ into two parts: (1) the
primary strip and (2) the open strip.
Figure~\ref{fig:floor-of-inside-case} illustrates the paths $F^*$ and
$U^*$, and Figure~\ref{fig:inside-strip-types} shows the schematic view
of the primary and the open strips.
All types of strips and components will be defined later.
The lowest-floor is constructed in such a way that it includes
all the exits of a component $C$.
We thus assume that the sink $t$ is a vertex on the boundary of $C$,
and there is no other exit.

\begin{figure}
  \centering
  \includegraphics{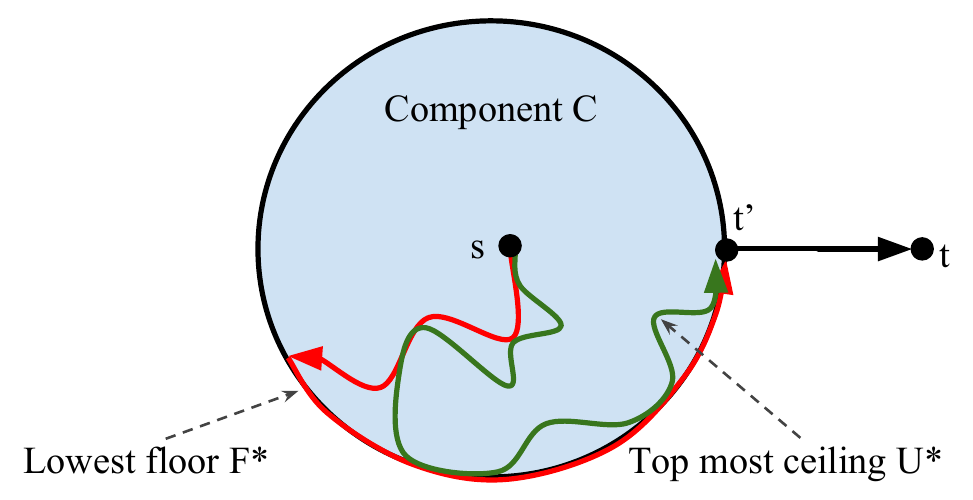}
  \caption{the lowest-floor path $F^*$ and the top-most ceiling $U^*$
           in the inside case.}
  \label{fig:floor-of-inside-case}
\end{figure}

\begin{figure}
  \centering
  \includegraphics{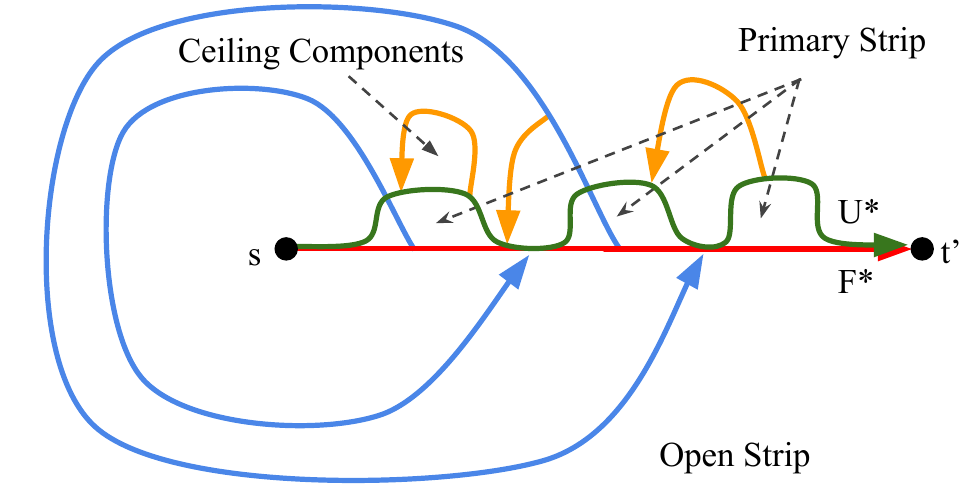}
  \caption{The types of major strips in the inside case.}
  \label{fig:inside-strip-types}
\end{figure}

\subsubsection{The lowest-floor path $F^*$.}
\label{sec:s-inside:lowest-floor}
We mainly apply the algorithm from Section~\ref{sec:st-outside} to the
primary strip defined by two paths $F^*$ and $U^*$, which may
not be arc-disjoint.  The path $F^*$ is the {\em lowest-floor path}
defined similarly to the outside case: We compute $F^*$ by finding
the right-most path $P$ from $s$ to $t$ (hence, $P$ is an $s,t$-path).
We recall that the sink $t$ is placed on the unbounded face of
the planar embedding.
Thus, $t$ must lie outside of the region enclosed by $C$.
This means that the right-most path $P$ has to go through
some exit $t'$ on the boundary of $C$.
Although we know that the right-most path goes from $t'$ to $t$ directly,
we detour the path $P$ along the boundary to include all the exits,
which is possible since the source $s$ is enclosed inside the component.
See Figure~\ref{fig:finding-lowest-floor-path} for illustration.

\begin{figure}
  \centering
  \includegraphics{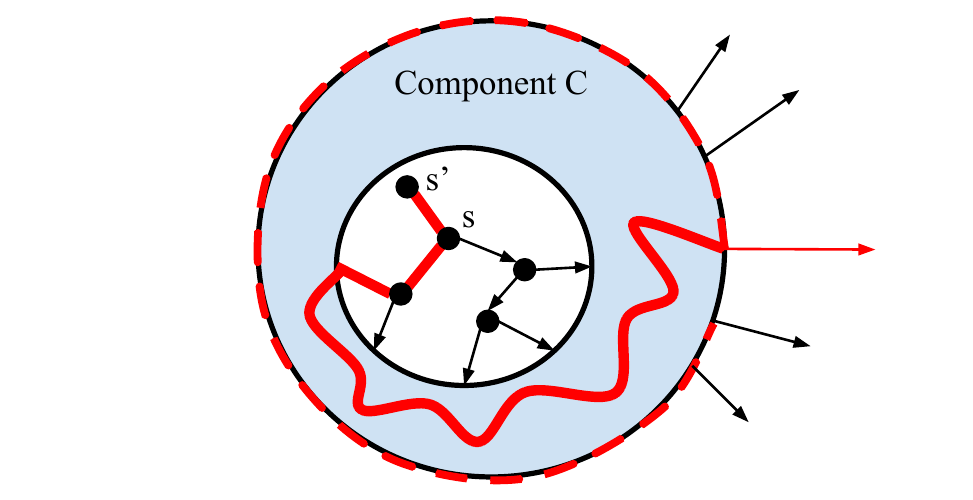}
  \caption{How to find $F^*$: the path $P$ is shown in solid red. The
    extension on the boundary is shown as a dashed line.}
  \label{fig:finding-lowest-floor-path}
\end{figure}

\subsubsection{Flipped paths.}
Unlike the outside case, the path $F^*$ in the inside case
does not divide the component $C$ into two pieces because the source
vertex $s$ is inside some inner face of $C$.
This leaves us more possibilities of paths that we have not
encountered in the outside case, called ``flipped paths''.
These are paths that goes from the left-side of the floor $F^*$
to the right-side (or vice versa).
(Please recall the terminology of left and right directions
in Section~\ref{sec:prelim:left-right}, and
recall the types of paths in Section~\ref{sec:prelim:forward-backward}.)
Fortunately, the reverse-flipped and hanging paths
do not exist because of our choice of
the lowest floor path $F^*$.

\begin{lemma}
\label{lem:no-reserse-flip-nor-hanging}
Consider the construction of the lowest floor path $F^*$
as in Section~\ref{sec:s-inside:lowest-floor}.
Then there is no reverse-flipped nor hanging paths w.r.t. $F^*$.
\end{lemma}
\begin{proof}
Consider a connected component $C$ that encloses the source $s$.
We recall that $F^*$ is constructed by finding the right-most $s,t$-path
and then detouring along the boundary of $C$.
Thus, $F^*$ is a union of two subpaths:
(1) the path $F_1$ that is a subpath of the right-most $s,t$-path
and
(2) the path $F_2$ that is a subpath of the boundary of $C$.

First, we rule out the existence of a hanging path.
Assume to the contrary that
there is a path $P\subseteq C$ that leaves $F^*$ to the right and
enters $F^*$ again from the right.
Then clearly $P$ must start and end on the first subpath $F_1$ of $F^*$
(otherwise, $P$ would have left the boundary).
But, this would contradict the fact that $F_1$ is a subpath of
the right-most $s,t$-path.

Next, we rule out the existence of a reverse-flipped path.
Assume to the contrary that
there is a path $P$ that leaves $F^*$ to the right and
enters $F^*$ again from the left.
We assume the minimality of such path
(thus, $P$ contains no internal vertices of $F^*$).
Then we know that $P$ must start from $F_1$.
(otherwise, $P$ would have left the boundary).
If $P$ is a backward path,
then the path $P$ together with a subpath of $F^*$ forms a clockwise cycle,
a contradiction
(we would have gotten ride of them in the preprocessing step).
If $P$ is a forward path, then it would contradict the
fact that $F_1$ is a subpath of the right-most path
since we can traverse along the path $P$ (which is right to $F_1$)
to get an $s,t$-path.
\end{proof}

\subsubsection{Primary Strip.}
The {\em primary strip} is defined by two paths, which may not be
arc-disjoint, namely $F^*$ and $U^*$.  The path $F^*$ is the lowest-floor
path.  The path $U^*$ is the {\em top-most ceiling}, which is an
$s,t$-path such that the strip $C_{U^*,F^*}$ encloses all
(non-flipped) forward paths (w.r.t. $F^*$). 
The path $U^*$ is essentially the left-most path from $s$ to $t$.  Section~\ref{sec:s-inside:init} describes how we find $U^*$ given $F^*$.

\subsubsection{Open Strip.}
The {\em open strip} $\hat{C}$ is defined to be everything
not in the inner parts of the primary strip.
See Figure~\ref{fig:inside-path-types} for illustration.

\begin{figure}
  \centering
  \includegraphics{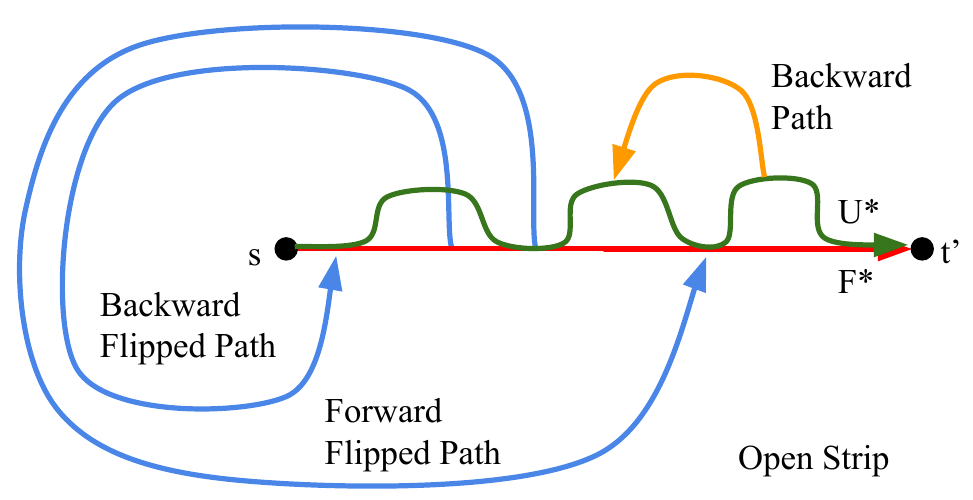}
  \caption{The types of paths in the open strip.}
  \label{fig:inside-path-types}
\end{figure}

To deal with arcs in $\hat{C}$, we first characterize the first set of
useless arcs: ceiling arcs.  An arc $e\in\hat{C}=(u,v)$ is a {\em
  ceiling arc} if $u$ can be reachable only from paths leaving $F^*$
to the left and $v$ can reach $t$ only through paths entering $F^*$
from the left as well.  These arcs form ``ceiling components''.  (Note
that paths that form ceiling components are all backward paths because
all the forward paths w.r.t. $F^*$ have been enclosed in
$C_{U^*,F^*}$.)  We shall prove that these paths are useless in
Lemma~\ref{lem:backward-comp-useless}.  Therefore, as a preprocessing
step, we delete all the ceiling arcs from $\hat{C}$.

To deal with the remaining arcs in $\hat{C}$, we first find strongly
connected components $H_1,H_2,\ldots$ in $\hat{C}$.
The arcs outside strongly connected components then form a directed
acyclic graph $\hat{D}$ that represents essentially the structures of
all flipped paths.  We can now process each component $H_i$ using the
outside-case algorithm only if we know which adjacent arcs in $\hat{D}$
are valid entrances and exits.  For arcs in $\hat{D}$ that belong to
some flipped path, we can use a dynamic programming algorithm to
compute all reachability information needed to determine their
usefulness.

Our algorithm can be described shortly as follows.
\begin{itemize}
\item {\bf Initialization.}
  We compute the lowest-floor path $F^*$ and the top-most ceiling
  $U^*$, thus forming the primary strip $C_{U^*,F^*}$ and the open
  strip $\hat{C}$.
\item {\bf The open strip (1): structures of flipped paths.}
  We find all strongly connected components in $\hat{C}$ and collapse
  them to produce the directed acyclic graph $\hat{D}$.
  We then compute the reachability information of arcs in $\hat{D}$ and
  use it to determine the usefulness of arcs in $\hat{D}$.

\item {\bf The open strip (2): process each strongly
  connected component.}
  Then we apply the outside-case algorithm to each strongly connected
  component in $\hat{C}$.
\item {\bf The primary strip.}
   We determine the usefulness of arcs in the primary strip.
\end{itemize}
In the following subsections, we describe how to deal with each part
of $C$.


\subsection{Initialization: Finding the Primary and Open Strips}
\label{sec:s-inside:init}

In this section, we discuss the initialization step of our algorithm,
which computes the primary and the open strips.
In this step, we also compute the reachability information of flipped
paths, which will be used in the later part of our algorithm.

As discussed in the previous section,
we compute the lowest-floor path $F^*$ by finding the right-most
$s,t$-path, which is unique and well-defined because
$s$ has a single arc leaving it.
The path $F^*$ is contained in $C$ because
we assume that $t$ is an exit vertex on the boundary of $C$.

The top-most ceiling $U^*$ can be computed by simply
computing the left-most non-flipped path w.r.t. of $F^*$
from $s$ to $t$.
This can be done in linear time by first
(temporary) removing every arc on the right-side of $F^*$
and then computing the left-most path from $s$ to $t$.
Consequently, we have the primary strip $C_{U^*,F^*}$,
we claim that it encloses all forward paths (w.r.t. $F^*$).

\begin{lemma}
\label{lem:main-strip-encloses-all-forward}
Every forward path w.r.t. $F^*$ is enclosed in $C_{U^*,F^*}$.
\end{lemma}
\begin{proof}
Suppose to the contrary that there exists a forward path $P$
not enclosed by $C_{U^*,F^*}$.
Then $P$ must intersect with $U^*$ or $F^*$.
If it is the former case, then $P$ must have a subpath $Q$
that leaves and enters $U^*$ on the left
(this is because $U^*$ is on the left of $F^*$).
We choose a subpath $Q$ in which the start vertex of $Q$,
say $u$, appears before its end vertex, say $v$, on $U^*$.
Such path $Q$ must exist because, otherwise, $P$ would have
a self-intersection inside $C_{U^*,F^*}$.
We may further assume that $Q$ is a minimal such path, which
means that all the arcs of $Q$ lie entirely on the left of $U^*$.
But, then $U^*$ would not be the left-most non-flipped path
(w.r.t. $F^*$) because the left-first-search algorithm
would have followed $Q$ instead of the $u,v$-subpath of $U^*$,
a contradiction.

The case that $P$ intersects $F^*$ is similar, and such path
would contradict the fact that $F^*$ is the right-most path.
Therefore, all the forward paths (w.r.t. $F^*$) must be enclosed in
$C_{U^*,F^*}$.
\end{proof}

Next consider the remaining parts of the component, which
forms the open strip $\hat{C}=C\setminus C_{U^*,F^*}$.
We have as a corollary of Lemma~\ref{lem:main-strip-encloses-all-forward}
that there is no simple $s,t$-path in the open strip that leaves $F^*$
from the right.

\begin{corollary}
\label{cor:no-right-path}
No simple $s,t'$-path that lies on the right of $F^*$.
\end{corollary}
\begin{proof}
Suppose there is an $s,t$-path $P$ in the open strip that leaves
$F^*$ from the right.
Then we have three cases.
First, if $P$ also enters $F^*$ from the right,
then $P$ cannot be a forward-path and cannot have any
forward subpath by Lemma~\ref{lem:main-strip-encloses-all-forward}.
Thus, $P$ must have a self-intersection, contradicting to the
fact that $P$ is a simple path.
Second, if $P$ enters $F^*$ from the left,
then we can find a $u,v$-subpath $Q$ of $P$ such that $u$ appears
before $v$ in $F^*$ and $Q$ is arc-disjoint from $F^*$.
But, then the right-first-search algorithm would have followed $Q$
instead of  the $u,v$-subpath of $F^*$, a contradiction.
Otherwise, if $P$ never enters $F^*$, then it must go directly to $t$.
But, then the right-first-search algorithm would have followed
the subpath of $P$ that leaves $F^*$, a contradiction.
\end{proof}


\subsection{Working on the Open Strip (1): Dealing with Arcs in Flipped Paths}
\label{sec:arcs-in-flipped-paths}

To determine the usefulness of arcs on the DAG $D$, 
which is formed by contracting strongly connected components 
$H_1,H_2,\ldots$ in the open strip, 
it suffices to check if it is contained in some useful flipped path
w.r.t. $F^*$.

We first prove the characterization of useful flipped paths.
Section~\ref{sec:s-inside:reachability} describes how we can 
check if an arc belongs to any useful path
by using reachability information.
Our first observation is that every forward flipped path is useful.

\begin{lemma}
\label{lem:flip-forward-is-useful}
Every forward flipped path w.r.t. $F^*$ is useful.
\end{lemma}
\begin{proof}
Consider any forward flipped path $P$.
By the definition of the forward flipped path,
$P$ shares no arcs with $F^*$.
Moreover, $P$ starts at some vertex $v\in F^*$ and ends at some vertex
$w\in F^*$ such that $v$ appears before $w$ in $F^*$.
Note that $F^*$ is a path that goes from $s$ to the vertex $t'$ in the
boundary of the component $C$.
We extend $F^*$ to a simple $s,t$-path $R$ by adding a $t',t$-path
that is not in the component $C$.
We then replace the $u,v$-subpath of $R$ by $P$, thus getting
a new path $R'$, which is a simple path because
$R$ shares no arcs with $P$ and contains no vertices of
$V(P)\setminus\{v,w\}$.
Thus, $R'$ is a simple $s,t$-path, implying that $P$ is useful.
\end{proof}

If an arc is not contained in any forward flipped path, it must be in
some backward flipped path.  In most cases, a backward flipped path is
useless except when there exists an exit that provides an alternative
route to $t$.

Consider the primary strip, which is enclosed by
the lowest-floor $F^*$ and the top-most ceiling $U^*$.
There are some arcs that $U^*$ and $F^*$ have in common.
If we remove all of these arcs, then the remaining graph
consists of weakly connected components, which we call humps.
Formally, a {\em hump} is a strip
obtained from the primary strip by removing $F^*\cap U^*$.
We can also order humps by the ordering of
their first vertices on the floor $F^*$.
The structure of humps is important in determining the usefulness of
a flipped path.
See Figure~\ref{fig:useful-flipped-backward-path}.

\begin{figure}
  \centering
  \includegraphics{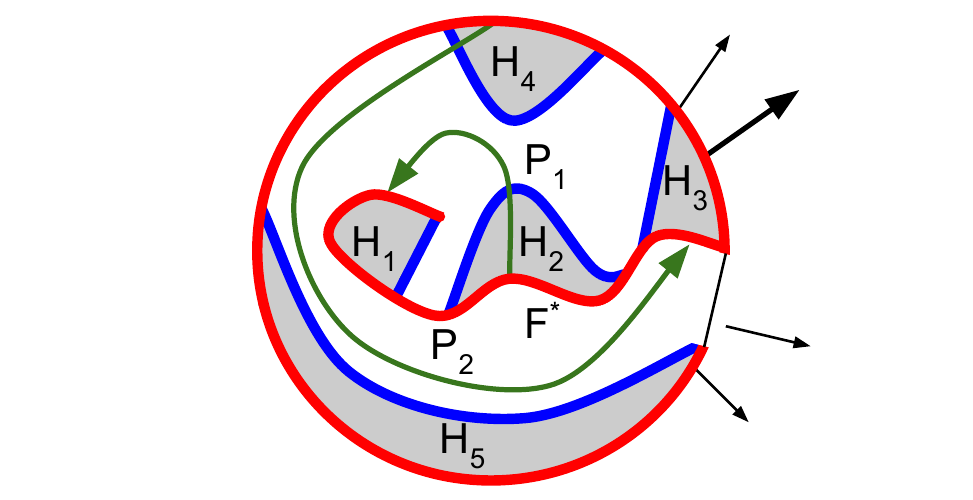}
  \caption{An example of flipped backward paths $P_1$ (useless) and
    $P_2$ (useful).  $F^*$ is shown in red and $U^*$ is shown in blue.
    There are 5 humps (shown in gray).}
  \label{fig:useful-flipped-backward-path}
\end{figure}

The lemma below shows the conditions when arcs in a backward flipped
path can be useful.

\begin{lemma}
\label{lem:useful-backward-flip-paths}
Let $P$ be any backward flipped-path such that $P$ starts at a vertex
$u\in F^*$ and ends at a vertex $v\in F^*$.
Then $\hat{P}=P\cap\hat{C}$ (the subpath of $P$ in the open strip)
is useful iff
$P$ ends at some hump $H$ (thus, $H$ contains $v$) such that either
\begin{enumerate}
\item both $u$ and $v$ are in the same hump $H$, and
      $v$ is not the first vertex in $H$.
\item The vertex $u$ is not in the hump $H$, and
      there is an exit after $v$ in the hump $H$.
\end{enumerate}
\end{lemma}
\begin{proof}
  First, we show that if the end vertex $v$ of $P$
  is not in any hump, then $\hat{P}=P\cap\hat{C}$ is useless.
  We assume wlog that the start vertex $u$ of $P$ is also
  not contained in any hump;
  otherwise, we contract all vertices in the hump
  (that contains $u$) to the vertex $u$.
  Thus, $P=\hat{P}$.
  It suffices to show that there is no $s,u$-path $J$
  that is vertex disjoint from $P$, which will imply that
  $P$ cannot be a useful path.
  To see this, assume a contradiction that there is an $s,u$-path $J$
  that is vertex disjoint from $P$.
  Then $J$ has to leave $F^*$ at some vertex $x$
  and then enters $F^*$ again at some vertex $y$
  (it is possible that $y=u$) in such a way that
  $x,y,u,v$ appear in the order $(x,v,y,u)$ on $F^*$.
  Let $J'$ be the $x,y$-subpath of $J$.
  If $J'$ is contained in the primary strip,
  then $J'$ would induce a hump containing $v$, a contradiction.
  So, we assume that $J'$ is not contained in the primary strip.
  If $J'$ leaves $F^*$ from the right,
  then $J'$ cannot enter $F^*$;
  otherwise, it would contradict the fact that $F^*$ is
  the right-most path.
  Hence, we are left with the case that $J'$ leaves $F^*$
  from the left and then enters again from the right.
  Since $J$ and $P$ are vertex-disjoint (and so are $J'$ and $P$),
  $J'$ cannot enter $F^*$ at any vertex that appears after $v$
  (including $u$ and $y$) because, otherwise, $J'$ has to cross the
  path $P$.
  Thus, we again have a contradiction.
  Consequently, since there is no path from $s$ that can reach
  $u$ or the hump containing $u$ without intersecting with $\hat{P}$,
  the path $\hat{P}$ cannot be a useful path.

  Next consider the case that $v$ is contained in some hump $H$.
  Let $w$ be a vertex in $P\cap U^*$, i.e., $w$ be a vertex in the
  ceiling intersecting with $P$.
  Also, let $\hat{P}$ be $P\cap\hat{C}$.

  For the backward direction, we construct
  a useful path $Q$ containing $\hat{P}$ by taking the union of the
  prefix of the top-most ceiling $U^*$ up to $w$, $\hat{P}$ and
  the subpath $P_F$ of the floor $F^*$ after $v$.
  In Case (1), $P_F$ is the suffix of $F^*$ after $v$, and
  in Case (2), $P_F$ starts from $v$ and ends at the exit.
  It can be seen that $Q$ is a simple path from $s$ to an exit,
  meaning that $Q$ is a useful path and so is $\hat{P}$.

  Now consider the forward direction.
  We proceed by contraposition.
  There are two subcases we have to consider:
  (i) $u$ and $v$ are in the same hump $H$, but
  $v$ is the first vertex in the hump and
  (ii) $u$ is not in the hump $H$ and there is no exit after
  $v$ in $H$.

  In Subcase (i), any useful path $Q$ containing $\hat{P}$ whose
  prefix up until $\hat{P}$ is entirely in the primary strip has to go
  through $v$; therefore, it cannot be simple.
  Consider the case when there is a useful path $Q$ containing
  $\hat{P}$ outside the primary strip.
  The prefix $Q'$ of this path must reach $w$ without touching $v$
  (because $v\in \hat{P}$).
  Observe that the union of $H$ and $\hat{P}$ contains a cycle
  enclosing the source $s$,
  and by the construction of $U^*$, there is no path from $s$ entering
  $H\cap U^*$ (otherwise, it would have been included in $U^*$).
  Consequently, because of planarity, the prefix $Q'$ must cross
  $\hat{P}$, meaning that $Q$ cannot be simple.

  Finally, we deal with Subcase (ii).
  Let $x$ be the last floor vertex of $H$.
  First note that, by the choice of $F^*$, any path from $s$ to $w$
  not intersecting $\hat{P}$ cannot cross $F^*$ to the right;
  therefore, it has to use $x$.
  Because of the same reason, any path from $v\in F^*$ to $t$ cannot
  leave $F^*$ to the right; thus, it must also go though $x$ as well.
  Consequently, a path from $s$ to $t$ containing $\hat{P}$ in this case
  cannot be simple.
\end{proof}

\subsubsection{Reachability information}
\label{sec:s-inside:reachability}

After we compute the primary strip and the open strip, and collapse
strongly connected components in $\hat{C}$, we run a linear-time
preprocessing to compute reachability information.

For each vertex $u$ in $\hat{D}$, we would like to compute $\first(u)\in
F^*$ defined to be the first vertex $w$ in $F^*$ such that there is a
path from $s$ to $u$ that uses $w$ but not other vertex in $F^*$ after
$w$.  We also want to compute $\last(u)\in F^*$ defined to be the last
vertex $w'$ in $F^*$ such that there is a path from $u$ to $t$ that
uses $w'$ as the first vertex in $F^*$.

To compute $\first(u)$, 
we find the left-first-search tree $T$ rooted at $s$.
For $u\in\hat{C}$, we set $\first(u)$ to be the closest ancestor of $u$ 
on the floor $F^*$.  
We compute $\last(u)$ similarly by 
finding the left-first-search tree in the reverse graph rooted at $t$
and setting $\last(u)$ as the closest ancestor of $u$ 
on the floor $F^*$.


\subsubsection{Checking arcs in $\hat{D}$}

Now we describe our linear-time algorithm for determining
the usefulness of arcs in $\hat{D}$

Consider an arc $e=(u,v)\in\hat{D}$.  If $\phi(\first(u))>\phi(\last(v))$,
then there exist a flipped forward path containing $e$.
Thus, by Lemma~\ref{lem:flip-forward-is-useful}, $e$ is useful.

If $\phi(\first(u))\leq\phi(\last(v))$, then we need to check conditions in
Lemma~\ref{lem:useful-backward-flip-paths}.  To do so, we need to
maintain additional information on every vertex $v$ in $F^*$, namely
the hump $H(v)$ that contains it, a flag representing if $v$ is the
first vertex in the hump, and a flag representing if the hump $H(v)$
contains an exit.  Since these data can be preprocessed in linear-time,
we can perform the check for each arc $e$ in constant time.

\subsection{Working on the Open Strip (2):
Dealing with Arcs in Strongly Connected Components}
\label{sec:arcs-in-scc}

Consider a strongly connected component $H_i$ in the open strip.  If
no arcs going into $H_i$ are useful or no arcs leaving $C$ are useful,
clearly every arc in $H_i$ are useless.  However, it is not obvious
that applying the previous outside algorithm simply works because
entering and leaving arcs determine the usefulness of the path.  See
Figure~\ref{fig:scc-in-open-strip-dependency}, for example.

\begin{figure}
  \centering
  \includegraphics{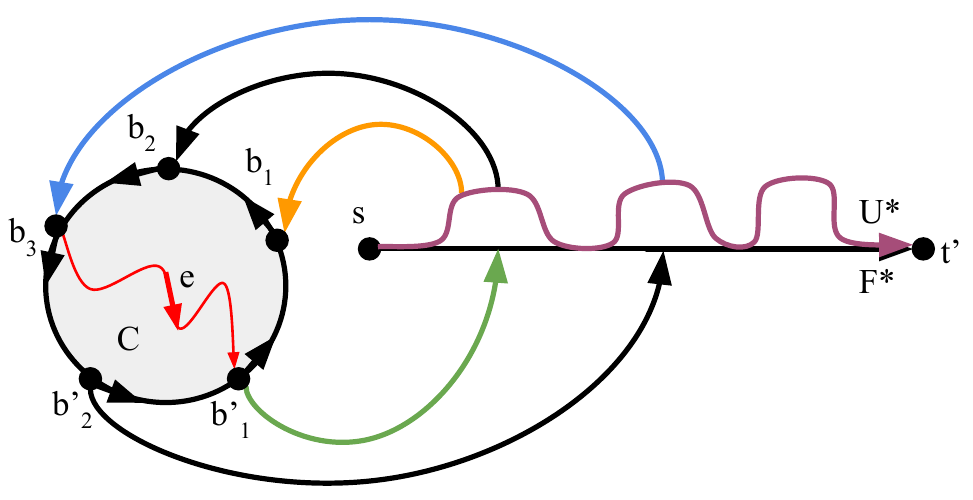}
  \caption{The arc $e$ in the red path is useful when considering $C$
    as an outside case instance.  However, using the blue path as
    an entering path (at the entrance $b_3$) and the green path as
    a leaving path (at the exit $b'_2$) is not enough to show that.
    One can, instead, use the orange path (at the entrance $b_1$)
    as an entering path together with parts of the boundary to
    construct a forward flipped path that contains $e$.}
  \label{fig:scc-in-open-strip-dependency}
\end{figure}

The following lemma proves that we can take the strongly connected
component $H_i$ and build an outside case instance by attaching, for
each useful arc in the DAG entering $H_i$ or leaving $H_i$,
corresponding to an entrance or an exit.

\begin{lemma}
Every arc in an outside instance of $H_i$ is useful iff it is
useful in the original graph.
\end{lemma}
\begin{proof}
  The backward direction is straight-forward.  Consider the forward
  direction: we would like to show that if an arc $e$ is useful in the
  instance, there exists a useful path containing it.

  Suppose that $H_i$ has $n_1$ entrances and $n_2$ exits.
  From Lemma~\ref{lem:boundary-is-ccw} and the observation
  in the proof of Lemma~\ref{lem:one-useful} that
  entrance-exit pairs do not appear interleaving,
  we know that $H_i$ is enclosed by a counter-clockwise
  cycle and we can name all entrances as $b_1,b_2,\ldots,b_{n_1}$ and
  all exits as $b'_1,b'_2,\ldots,b'_{n_2}$ in such a way that they
  appear in counter-clockwise order as
  \[
  b_1,b_2,\ldots,b_{n_1},b'_1,b'_2,\ldots,b'_{n_2}.
  \]

  Consider an arc $e$ which is useful in this instance.  Since $e$ is
  useful, there is a simple path $P$ from some entrance $b_i$
  to some exit $b'_j$ containing $e$.  Since the boundary arcs form a
  counter-clockwise cycle, we can extend $P$ to start from $b_1$ by
  adding boundary arcs from $b_1$ to $b_i$.  We can also extend $P$ to
  reach $b'_{n_2}$ by adding arcs from $b'_j$ to $b'_{n_2}$.

  Since an arc entering $C$ at $b_1$ is useful and an arc leaving $C$
  from $b'_{n_2}$ is also useful, we can construct a useful path
  containing $e$ by joining a path from $s$ to $b_1$, the path $P$,
  and a path from $b'_{n_2}$ to $t$. Thus, $e$ is useful.
\end{proof}


\subsection{Working on the Primary Strip}
\label{sec:s-inside:primary}

Given the primary strip $C_{U^*,F^*}$, we can use the algorithm from
Section~\ref{sec:st-outside} to find all useful arcs inside
$C_{U^*,F^*}$.  The next lemma proves the correctness of this step.

\begin{lemma}
  A useful arc $e$ in $C$ is useful iff it is useful in $C_{U^*,F^*}$.
\end{lemma}
\begin{proof}
  The backward direction is obvious.  We focus on the forward
  direction.  Assume for contradiction that there exist a useful arc
  $e$ in $C$ but $e$ is not useful in $C_{U^*,F^*}$.  Consider a
  useful path $P$ from $s$ to $t$ containing $e$.  Since $e$ is inside
  $C_{U^*,F^*}$, $P$ must cross the lowest-floor $F^*$ or the top-most
  ceiling $U^*$.  However, $P$ cannot cross $U^*$ as it would create a
  forward path outside the primary strip (contradicting
  Lemma~\ref{lem:main-strip-encloses-all-forward}).  Now suppose that
  $P$ crosses $F^*$ at some vertex $w$ before $e$.  If $P$ do not
  leave the primary strip after $w$, then clearly $e$ must be useful inside
  the primary strip.  If $P$ leaves the primary strip at the ceiling
  after $e$, the subpath of $P$ containing $e$ is again a forward
  path.  We also note that $P$ cannot leave $C_{U^*,F^*}$ at the floor
  because it would cross itself.  Since we reach contradiction in
  every case, the lemma follows.
\end{proof}

\subsection{The Ceiling Components are Useless}
\label{sec:ceiling-components}

In this section, we show that every arc in the ceiling components is useless.
The ceiling components are formed by arcs $e=(u,v)$ in $\hat{C}$ which
are reachable only from paths leaving $F^*$ to the left and $v$ can
reach $t$ only through paths entering $F^*$ from the left.

\begin{lemma}
  Every ceiling arc $e$ is useless.
  \label{lem:backward-comp-useless}
\end{lemma}
\begin{proof}
  First, note that $e$ does not lie in a forward path; otherwise $e$ 
  would be in the primary strip.

  Let $P$ be a useful path from $s$ to $t$ containing $e$.  Let $u'$
  be the last vertices of $P\cap U^*$ before reaching $e$ and $v'$ be
  the first vertices of $P\cap U^*$ after leaving $e$.

  Because of the degree constraint, $P'$ must use the only incoming
  arc of $u'$, say $e_u$, and the only outgoing arc of $v'$,
  say $e_v$.  Let $P_1$ be the prefix of $P'$ from $s$ to
  the head of $e_u$ and $P_2$ be the suffix of $P'$ from the tail of $e_v$
  to $t$.  The only way $P_2$ can avoid crossing $P_1$ is to
  cross $U^*$; however, this creates a clockwise cycle, a contradiction.
\end{proof}


\section{Conclusions and Open Problems}

In this paper, we presented the algorithm that simplifies
a directed planar network into a plane graph in which 
every vertex except the source $s$ and the sink $t$ has degree three, 
the graph has no clockwise cycle,
and every arc is contained in some simple $s,t$-path.
Our algorithm can be applied as a preprocessing step
for Weihe's algorithm \cite{Weihe97} and thus yields 
an $O(n\log n)$-time algorithm for computing
maximum $s,t$-flow on directed planar graphs.
This gives an alternative approach that departs from 
the $O(n\log n)$-time algorithm
of Borradaile and Klein \cite{BorradaileK09} and 
that of Erickson \cite{Erickson10},
which are essentially the same algorithm with different interpretations.
While other works mainly deal with maximum flow,
the main concern in this paper is in simplifying the flow network.
Henceforth, we believe that our algorithm will serve as
supplementary tools for further developments of network flow algorithms.

Next let us briefly discuss open problems.
A straightforward question is whether 
our approach can be generalized to 
a larger class of graphs on surfaces.
Another problem that might be interesting for readers
is to remove the prerequisites from our main algorithm.
Specifically, prior to feeding a plane graph to the algorithm,
we apply a sequence of reductions 
so that each vertex has degree three, 
and the plane graph contains no clockwise cycles.
These are the prerequisites required by our main algorithm.
Although the reduction does not change the value of the maximum flow,
the usefulness of arcs in the modified network 
may differ from the original graph.
It would be interesting to simplify the flow network without
changing the usefulness of arcs in the graph.

Lastly, we would like to note that if there exists a reduction  
that deals with clockwise cycles, 
then the degree requirement can be removed.
Specifically, if there is 
a procedure that given a plane graph $G$ constructs 
a new plane graph $G'$ with no clockwise cycles together
with another efficient procedure for identifying the usefulness of
original arcs in $G$ based on the results in $G'$, 
then we can apply the following lemma.

\begin{lemma}
\label{lem:no-cw-cycles}
  If a plane graph $G$ contains no clockwise cycles, a new plane graph
  $G'$ constructed by replacing every vertex of degree greater than three
  in $G$ with a {\bf clockwise cycle} preserves the usefulness of 
  every arc from $G$ (those arcs that are not contained in any clockwise cycle).
\end{lemma}
\begin{proof}
We first describe the reduction formally.
Let $G$ be a plane graph with no clockwise cycle,
and let $s$ and $t$ be the source and sink vertices.
For each vertex $v\in V(G)$, let $d_v$ denote the degree of $v$ in $G$.
We construct $G'$ by first adding copies of $s$ and $t$, 
namely $s'$ and $t'$, respectively.
Then we add to $G'$ a clockwise cycle $C_v$ on $d_v$ vertices, 
for each vertex $v\in V(G)\setminus\{s,t\}$.
Each vertex $v_e$ in the cycle $C_v$ corresponds to an arc $e$
incident to $v$ in $G$, and vertices in $C_v$ are sorted in 
the same cyclic order as their corresponding arcs in $G$.
Next we add an arc $u_{uv}v_{uv}$ to $G'$ for every arc $uv\in E(G)$.
Observe that $G'$ is a planar graph obtained by replacing
each vertex of $G$ by a cycle,
and we can keep the same planar drawing as that of $G$.
In particular, the resulting graph $G'$ is a plane graph.

It can be seen that 
every simple path in $G$ corresponds to some simple path in $G'$.
Thus, every useful arc in $G$ (w.r.t. $s$ and $t$) is 
also useful in $G'$ (w.r.t. $s'$ and $t'$).
Now let us consider a useless arc $uv$ in $G$ (w.r.t. $s$ and $t$).
Assume to the contrary that there is a simple $s',t'$-path $P'$ in $G'$
containing the arc $v_{uv}u_{uv}$,
which is the arc corresponding to $uv$.
Since $G$ has no simple $s,t$-path containing $uv$,
we know that $P'$ maps to a walk $P$ in $G$ that visits some vertex $w$
at least twice.
Thus, $P'$ must visit the cycle $C_w$ at least twice as well.
Let us say $P'$ enters $C_w$ at a vertex $w_{e_1}$, 
leaves $C_w$ from a vertex $w_{e_2}$,
enters $C_{w}$ again at a vertex $w_{e_3}$
and then leaves $C_{w}$ from a vertex $w_{e_4}$.
Then we can construct a cycle $Q$ by walking along 
the $w_{e_1},w_{e_3}$-subpath of $P'$ 
and then continue to the $w_{e_3},w_{e_1}$-subpath of $C_w$.
Since the original graph $G$ has no clockwise cycle, 
$Q$ must be counterclockwise.
Moreover, since $C_w$ is a clockwise cycle,
the vertex $w_{e_4}$ must lie between $w_{e_3}$ and $w_{e_1}$.
At this point, it is not hard to see that
the cycle $Q$ must enclose the $s',w_{e_4}$-subpath of $P'$.
So, the only way that $P'$ can leave the cycle $Q$ and reach $t'$
is to cross the $w_{e_2},w_{e_3}$-subpath of $C_w$.
But, this is not possible unless $P'$ crosses itself.
Hence, we arrive at a contradiction.

Therefore, the graph $G'$ preserves the usefulness of every arc from $G$.
\end{proof}

\medskip

\noindent{\bf Acknowledgement.}
We thank Joseph Cheriyan for introducing us the flow network simplification
problem. We also thank Karthik~C.S. for pointing out some typos.

Part of this work was done while Bundit Laekhanukit was visiting the Simons Institute for the Theory of Computing. It was partially supported by the DIMACS/Simons Collaboration on Bridging Continuous and Discrete Optimization through NSF grant \#CCF-1740425.

\bibliographystyle{alpha}
\bibliography{references}
\end{document}